\documentclass[letterpaper, 10 pt, conference]{ieeeconf}  
\IEEEoverridecommandlockouts                              
\usepackage{amsmath,amssymb,amsfonts}
\usepackage{textcomp}
\usepackage{bbm}
\usepackage{stmaryrd}
\usepackage{latexsym}
\usepackage{wrapfig}
\usepackage{thmtools}
\usepackage{thm-restate}

\usepackage{algorithm}
\usepackage{algorithmic}

\usepackage{subcaption}
\usepackage{color}
\usepackage{graphicx}

\allowdisplaybreaks[1]

\newcommand{\eqspace}[0]{\mathbin{\phantom{=}}}

\newcommand{\eqdef}{\ensuremath{\stackrel{\mathrm{def}}{=}}}

\newcommand{\argmin}{\operatornamewithlimits{argmin}}
\newcommand{\supp}{\mathsf{supp}}

\newcommand{\diam}{\mathsf{diam}}
\newcommand{\Dists}{\mathbb{D}} 
\newcommand{\Prob}{\mathrm{Pr}}

\newcommand{\lift}[1]{{#1}^{\#}}
\newcommand{\liftp}[1]{{#1}_{\Wp}^{\#}}
\newcommand{\liftemd}[1]{{#1}_{\EMD}^{\#}}

\newcommand{\nats}{\mathbb{N}}
\newcommand{\reals}{\mathbb{R}}
\newcommand{\realsnng}{\mathbb{R}^{\ge0}}
\newcommand{\realsone}{\mathbb{R}^{\ge1}}

\newcommand{\calf}{\mathcal{F}}

\newcommand{\cals}{\mathcal{S}}

\newcommand{\calx}{\mathcal{X}}
\newcommand{\caly}{\mathcal{Y}}
\newcommand{\calz}{\mathcal{Z}}

\newcommand{\DP}[0]{\textsf{DP}}
\newcommand{\XDP}[0]{\textsf{XDP}}

\newcommand{\KLP}[0]{\textsf{KLP}}

\newcommand{\DistP}[0]{\textsf{DistP}}
\newcommand{\XDistP}[0]{\textsf{XDistP}}

\newcommand{\utmetric}[0]{\mathit{d}}

\newcommand{\sensfunc}[0]{\mathit{W}}
\newcommand{\sensinf}[0]{\sensfunc_{\infty,\utmetric}}
\newcommand{\sensemd}[0]{\sensfunc_{1,\utmetric}}

\newcommand{\Div}[1]{\mathsf{Div}(#1)}
\newcommand{\diverge}[2]{\mathit{D}(#1 \parallel #2)}
\newcommand{\DKL}[0]{\mathit{D}_{\rm KL}}
\newcommand{\KLdiverge}[2]{\DKL(#1 \parallel #2)}

\newcommand{\DTV}[0]{\mathit{D}_{\rm TV}}
\newcommand{\Dinf}[0]{\mathit{D}_{\infty}}
\newcommand{\maxdiverge}[2]{\Dinf(#1 \parallel #2)}
\newcommand{\appmaxdiverge}[2]{\Dinf^{\delta}(#1 \parallel #2)}

\newcommand{\renyidiverge}[3]{\mathit{D}_{#3}(#1 \parallel #2)}

\newcommand{\Df}[0]{\mathit{D}_{f}}
\newcommand{\fdiverge}[2]{\renyidiverge{#1}{#2}{f}}

\newcommand{\Wp}{\mathit{W}_{p}}
\newcommand{\EMD}{\mathit{W}_{1}}

\newcommand{\EMDu}{\mathit{W}_{1,\utmetric}}
\newcommand{\Wpu}{\mathit{W}_{p,\utmetric}}
\newcommand{\Winfu}{\mathit{W}_{\infty,\utmetric}}

\newcommand{\cp}[2]{\mathsf{cp}(#1, #2)}
\newcommand{\GammaP}[0]{\mathit{\Gamma_{\!p,\utmetric}}}
\newcommand{\GammaInf}[0]{\mathit{\Gamma_{\!{\rm \infty,\utmetric}}}}
\newcommand{\GammaEMD}[0]{\mathit{\Gamma_{\!{\rm 1,\utmetric}}}}

\newcommand{\alg}{\mathit{A}}

\newcommand{\CP}{\mathit{C}}

\newcommand{\Seq}{\mathbin{\odot}}
\newcommand{\liftSeq}{\mathbin{\bullet}}

\newcommand{\hlambda}{\widehat{\lambda}}

\newcommand{\male}{{\it male}}
\newcommand{\female}{{\it female}}

\newtheorem{theorem}{Theorem}

\newtheorem{lemma}[theorem]{Lemma}

\newtheorem{definition}{Definition}
\newtheorem{example}{Example}

\newif\ifcommentson\commentsontrue

\newif\ifconferenceon\conferenceonfalse
\ifconferenceon
\newcommand{\arxiv}[1]{}
\newcommand{\conference}[1]{#1}
\newcommand{\conferenceShort}[1]{}
\newcommand{\journal}[1]{}
\else
\newcommand{\arxiv}[1]{#1}
\newcommand{\conference}[1]{}
\newcommand{\conferenceShort}[1]{}
\newcommand{\journal}[1]{}
\fi

\newcommand{\colorR}[1]{\textcolor{red}{#1}}

\ifcommentson
\newcommand{\commentsize}[0]{.45\textwidth}
\newcommand{\commentTM}[1]{\begin{center} \parbox{\commentsize}{\textbf{\textcolor{black}{Comment T.}} \textcolor{red}{#1 }}\end{center}}
\newcommand{\commentYK}[1]{\begin{center} \parbox{\commentsize}{\textbf{\textcolor{black}{Comment Y.}} \textcolor{red}{#1} }\end{center}}
\newcommand{\replyTM}[1]{\begin{center} \parbox{\commentsize}{\textbf{Reply T.} \textcolor{blue}{#1} }\end{center}}
\newcommand{\replyYK}[1]{\begin{center} \parbox{\commentsize}{\textbf{Reply Y.} \textcolor{blue}{#1} }\end{center}}
\newcommand{\commentT}[1]{\marginpar{\footnotesize \color{red} {\bf T:} \textsf{\scriptsize #1}}}
\newcommand{\commentY}[1]{\marginpar{\footnotesize \color{red} {\bf Y:} \textsf{\scriptsize #1}}}
\newcommand{\replyT}[1]{\marginpar{\footnotesize \color{red} {\bf T:} \textsf{\scriptsize #1}}}
\newcommand{\replyY}[1]{\marginpar{\footnotesize \color{red} {\bf Y:} \textsf{\scriptsize #1}}}
\else
\newcommand{\commentTM}[1]{}
\newcommand{\commentYK}[1]{}
\newcommand{\replyTM}[1]{}
\newcommand{\replyYK}[1]{}
\newcommand{\commentT}[1]{}
\newcommand{\commentY}[1]{}
\newcommand{\replyT}[1]{}
\newcommand{\replyY}[1]{}
\fi

\newcommand{\pagelimitmarker}[1]{~\\ {\colorR{\ifthenelse{\thepage>#1}{\Huge Exceeding the page limit}{\huge Within the page limit}}}~\\ {\huge{\colorR{Page Limit =  #1 ~~~ Current Page = $\thepage$}}}}

\title{\LARGE \bf
Local Distribution Obfuscation via Probability Coupling*
}

\author{Yusuke Kawamoto$^{1}$ \and Takao Murakami$^{2}$
\thanks{*This work was partially supported by JSPS KAKENHI Grant JP17K12667, JP19H04113, and Inria LOGIS project.}
\thanks{$^{1}$Yusuke Kawamoto is with AIST, Tsukuba, Japan.}%
\thanks{$^{2}$Takao Murakami is with AIST, Tokyo, Japan.}%
}

\begin{document}
\maketitle
\thispagestyle{empty}
\pagestyle{empty}

\begin{abstract}
We introduce a general model for the local obfuscation of probability distributions by probabilistic perturbation, e.g., by adding differentially private noise, and investigate its theoretical properties. Specifically, we relax a notion of distribution privacy (\DistP{}) by generalizing it to divergence, and propose local obfuscation mechanisms that provide divergence distribution privacy. To provide $f$-divergence distribution privacy, we prove that probabilistic perturbation noise should be added proportionally to the Earth mover's distance between the probability distributions that we want to make indistinguishable. Furthermore, we introduce a local obfuscation mechanism, which we call a \emph{coupling mechanism}, that provides divergence distribution privacy while optimizing the utility of obfuscated data by using exact/approximate auxiliary information on the input distributions we want to protect. 
\end{abstract}
\allowdisplaybreaks[1]

\section{Introduction}
\label{sec:introdction}
\emph{Differential privacy} (\DP{})~\cite{Dwork:06:ICALP} is one of the most popular privacy notions that have been studied in various areas, including databases, machine learning, geo-locations, and social networks. 
The protection of \DP{} can be achieved by 
adding probabilistic noise to the data we want to obfuscate.
In particular, many studies have proposed 
\emph{local obfuscation mechanisms}%
~\cite{Duchi:13:FOCS,Andres:13:CCS,Erlingsson_CCS14} 
that perturb each single ``point'' datum (e.g., a geo-location point) by adding controlled probabilistic noise before sending it out to a data collector.

Recent researches~\cite{Jelasity:IHMMSec:14,Kawamoto:19:ESORICS,Geumlek:19:ISIT} show that local obfuscation mechanisms can be used to hide the probability distributions 
that lie behind such point data 
and implicitly represent sensitive attributes (e.g., age, gender, social status).
In particular, \cite{Kawamoto:19:ESORICS} proposes the notion of \emph{distribution privacy} (\DistP{}) as the local \DP{} of probability distributions.
Roughly, \DistP{} of a local obfuscation mechanism $\alg$ represents that the adversary cannot significantly gain information on the distribution of $\alg$'s input by observing $\alg$'s output.
However, since \DistP{} assumes the worst case risk in the sense of \DP{}, it imposes strong requirement and might unnecessarily lose the utility of obfuscated data.

In this paper, we relax the notion of \DistP{} by generalizing it to an arbitrary divergence.
The basic idea is similar to 
point privacy notions that 
relax \DP{} and improve utility by relying on some divergence (e.g., total variation privacy~\cite{Barber:14:arXiv}, Kullback-Leibler divergence privacy~\cite{Barber:14:arXiv,Cuff:16:CCS}, 
and R\'enyi differential privacy~\cite{Mironov:17:CSF}).
We define the notion of \emph{divergence distribution privacy} by replacing the \DP{}-style with an arbitrary divergence $D$.
This relaxation allows us to formalize ``on-average'' \DistP{}, and to explore privacy notions against an adversary performing the statistical hypothesis test corresponding to the divergence \cite{Barber:14:arXiv}.

Furthermore, we propose and investigate local obfuscation mechanisms that provide divergence \DistP{}.
Specifically, we consider the following two scenarios: 
\begin{enumerate}\renewcommand{\labelenumi}{(\roman{enumi})}
\item when we have no idea on the input distributions;
\item when we know exact or approximate information on the input distributions
(e.g., when we can use public datasets \cite{Yang_TIST15,Yang:19:TKDE} to learn approximate distributions of locations of male/female users if we want to obfuscate the attribute male/female).
\end{enumerate}

For the scenario (i),
we clarify how much perturbation noise should be added to provide $f$-divergence \DistP{}
when we use an existing mechanism for obfuscating point data.
For the scenario (ii),
we introduce a local obfuscation mechanism that provides divergence \DistP{} while optimizing the utility of obfuscated data by using the auxiliary information.
Here it should be noted that 
probability coupling techniques are
crucial in constructing divergence \DistP{} mechanisms in both the scenarios.

\textbf{Our contributions.}~The main contributions are as follows:
\begin{itemize}
\item We introduce notions of divergence \DistP{} and investigate theoretical properties of distribution obfuscation,
especially the relationships between local distribution obfuscation and probability coupling.
\item We investigate the relationships among various notions of \DistP{} based on $f$-divergences, such as Kullback-Leibler divergence, which models ``on-average'' risk.
\item In the scenario (i), we present how much divergence \DistP{} can be achieved by local obfuscation.
In particular, by using probability coupling techniques, we prove that perturbation noise should be added proportionally to the Earth mover's distance between the input distributions that we want to make indistinguishable.
\item In the scenario (ii), we propose a local obfuscation mechanism, called a \emph{(utility-optimal) coupling mechanism}, that provides divergence \DistP{} while minimizing utility loss.
The construction of the mechanism relies on solving an optimal transportation problem using probability coupling.
\item We theoretically evaluate the divergence \DistP{} and utility loss of coupling mechanisms that can use exact/approximate knowledge on the input distributions.
\end{itemize}

\textbf{Paper organization.}~~
The rest of this paper is organized as follows.
Section~\ref{sec:preliminaries} presents 
background knowledge.
Section~\ref{sec:f-distDP} introduces notions of divergence \DistP{}.
Section~\ref{sec:properties} investigates important properties of divergence \DistP{}, and relationships among privacy notions.
Section~\ref{sec:DistP-by-DP-f} shows that in the scenario (i), an $f$-privacy mechanism can provide $f$-divergence \DistP{}.
Section~\ref{sec:DistP-auxiliary} generalizes \DistP{} to use exact/approximate information on the input distribution in the scenario (ii), and proposes a local mechanism for providing \DistP{} while optimizing utility.
Section~\ref{sec:related-work} discusses related work and
Section~\ref{sec:conclusion} concludes.
\conference{All proofs not appearing in Appendix are presented in the arXiv version~\cite{Kawamoto:Allerton19:arXiv}.}

\section{Preliminaries}
\label{sec:preliminaries}
In this section we recall some notions of privacy, divergence, and metrics used in this paper.

\subsection{Notations}
\label{subsec:notations}

Let $\realsnng$ be the set of non-negative real numbers,
and $[0, 1] \eqdef \{ r\in\realsnng \mid r \le 1 \}$.
Let 
$\varepsilon, \varepsilon_0, \varepsilon_1 \in \realsnng$, 
$\delta, \delta_0, \delta_1 \in [0, 1]$, 
and $e$ be the base of natural logarithm.

We denote by $|\calx|$ the number of elements in a finite set~$\calx$, and by $\Dists\calx$ the \emph{set of all probability distributions} over a set~$\calx$.
Given a probability distribution $\lambda$ over a finite set $\calx$, the probability of drawing a value $x$ from $\lambda$ is denoted by $\lambda[x]$.
For a finite subset $\calx'\subseteq\calx$, we define $\lambda[\calx']$ by $\lambda[\calx'] = \sum_{x'\in\calx'} \lambda[x']$.
For a distribution $\lambda$ over a finite set $\calx$, its \emph{support} is $\supp(\lambda) = \{ x \in \calx \colon \lambda[x] > 0 \}$.

For a randomized algorithm $\alg:\calx\rightarrow\Dists\caly$ and a set $R\subseteq\caly$, we denote by $\alg(x)[R]$ the probability that given an input $x$, $\alg$ outputs one of the elements of $R$.
For a randomized algorithm $\alg:\calx\rightarrow\Dists\caly$ and a distribution $\lambda$ over $\calx$, we define $\lift{\alg}(\lambda)$ as the probability distribution of the output of $\alg$.
Formally, the \emph{lifting} of $\alg:\calx\rightarrow\Dists\caly$ 
is the function $\lift{\alg}: \Dists\calx\rightarrow\Dists\caly$ 
such that for any $R\subseteq\caly$,
$
\lift{\alg}(\lambda)[R] \eqdef
\sum_{x\in\calx}\lambda[x] \alg(x)[R]
$.

\subsection{Differential Privacy}
\label{sub:DP}

\emph{Differential privacy} \cite{Dwork:06:ICALP} is a notion of privacy guaranteeing that we cannot learn which of two ``adjacent'' inputs $x$ and $x'$ is used to generate an output of a randomized algorithm.
This notion is parameterized by a degree $\varepsilon$ of indistinguishability, a ratio $\delta$ of exception, and some adjacency relation $\varPhi$ over a set $\calx$ of data.
The formal definition is given as follows.

\begin{definition}[Differential privacy] \label{def:DP} \rm
A randomized algorithm $\alg: \calx \rightarrow \Dists\caly$ provides \emph{$(\varepsilon,\delta)$-differential privacy (\DP{})} w.r.t. an adjacency relation $\varPhi\subseteq\calx\times\calx$ if for any $(x, x')\in\varPhi$ and any $R\subseteq\caly$,
\[
\Prob[ \alg(x)\in R ] \leq e^\varepsilon \,\Prob[ \alg(x')\in R ] + \delta
\]
where the probability is taken over the random choices in $\alg$.
\end{definition}

Then the protection of \DP{} is stronger for smaller $\varepsilon$ and $\delta$.

\journal{
One of the most important properties of \DP{} is the compositionality.
The sequential composition of \DP{} mechanisms also provides \DP{}:
For any $n$ independent randomized algorithms $\alg_1, \alg_2, \ldots, \alg_n$,
if each $\alg_i$ provides $(\varepsilon_i, \delta_i)$-\DP{}, 
then the sequential composition $\alg_n \circ \alg_{n-1} \circ \ldots \circ \alg_1$ provides $(\sum_{i=1}^{n}\varepsilon_i, \sum_{i=1}^{n} \delta_i)$-\DP{}.
}

\DP{} can be achieved by a \emph{local obfuscation mechanism} or  \emph{privacy mechanism} (illustrated in Fig.~\ref{fig:obfuscate-scenarios}), namely a randomized algorithm that adds controlled noise probabilistically to given inputs that we want to protect.
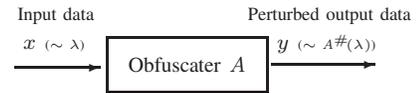
\begin{figure}[t]
\centering
\centering
\begin{picture}(250, 38)
 \put( 54, 27){\scriptsize Input data}
 \put( 56, 16){{\footnotesize $x$}~{\tiny ($\sim\lambda$)}}
 \put(141, 27){\scriptsize Perturbed output data}
 \put(152, 16){{\footnotesize $y$}~{\tiny ($\sim\lift{\alg}\!(\lambda)$)}}
 \linethickness{0.6pt}
 \put( 88,  0){\framebox(60,20){\footnotesize \mbox{Obfuscater $\alg$}}}
 \linethickness{0.8pt}
 \put( 53,  9){\vector( 1, 0){32}}
 \put(150, 10){\vector( 1, 0){39}}
\end{picture}
\caption{A local obfuscation mechanism $\alg$ perturbs input data $x$ and returns  output data $y$.
Then the underlying probability distribution $\lambda$ can also be seen to be obfuscated.}

\label{fig:obfuscate-scenarios}
\end{figure}

\subsection{Extended Differential Privacy (\XDP{})}
\label{sub:XDP}
The notion of \DP{} can be relaxed by incorporating a metric $d$ over the set $\calx$ of input data.
In~\cite{Chatzikokolakis:13:PETS} Chatzikokolakis \textit{et al.} propose the notion of ``$d$-privacy'', an extension of $(\varepsilon,0)$-\DP{} to a metric $d$ on input data.
Intuitively, this notion guarantees that when two inputs $x$ and $x'$ are closer in terms of $d$, the output distributions are less distinguishable%
\footnote{Compared to \DP{}, \XDP{} provides weaker privacy and higher utility, as it obfuscates closer points.
E.g., \cite{Alvim:18:CSF} shows the planar Laplace mechanism~\cite{Andres:13:CCS} (with \XDP{}) adds less noise than the randomized response (with \DP{}).
}.
Here we show the definition of this extended \DP{} equipped with $\delta$.
\begin{definition}[Extended differential privacy]\label{def:XDP-simple}\rm
Let $d: \calx\times\calx\rightarrow\reals$ be a metric.
We say that a randomized algorithm $\alg: \calx \rightarrow \Dists\caly$ provides \emph{$(\varepsilon,\delta,d)$-extended differential privacy (\XDP{})} if for all $x, x'\in\calx$ and $R\subseteq\caly$,
\begin{align*}
\Prob[ \alg(x)\in R ] \leq e^{\varepsilon d(x,x')} \,\Prob[ \alg(x')\in R ] + \delta
\end{align*}
where the probability is taken over the random choices in $\alg$.
\end{definition}

To achieve \XDP{}, obfuscation mechanisms should add noise proportionally to the distance $d(x,x')$ between the two inputs $x$ and $x'$ that we want to make indistinguishable, 
hence more noise is require for a larger $d(x,x')$.

\subsection{Distribution Privacy and Extended Distribution Privacy}
\label{sub:DistP-XDistP}

Distribution privacy (\DistP{})~\cite{Kawamoto:19:ESORICS} is a privacy notion that measures how much information on the input distribution is leaked by an output of a randomized algorithm.
For example, let $\lambda_{\male}$ (resp. $\lambda_{\female}$) be a (prior) probability distribution of the locations of the male (resp. female) users.
When we observe an output of an obfuscation mechanism $\alg$ and cannot learn whether the input to $\alg$ is drawn from $\lambda_{\male}$ or $\lambda_{\female}$, 
then we say that $\alg$ provides $(\varepsilon, \delta)$-\DistP{} w.r.t. $(\lambda_{\male}, \lambda_{\female})$.
Formally, \DistP{} is defined as follows.

\vspace{0.5ex}
\begin{definition}[Distribution privacy]\label{def:max-DistP}\rm
Let $\varepsilon\in\realsnng$ and $\delta\in[0,1]$.
We say that a randomized algorithm $\alg:\calx\rightarrow\Dists\caly$ provides \emph{$(\varepsilon,\delta)$-distribution privacy (\DistP{}) w.r.t.} 
an adjacency relation $\varPsi\subseteq\Dists\calx\times\Dists\calx$ 
if its lifting $\lift{\alg}:\Dists\calx\rightarrow\Dists\caly$ provides $(\varepsilon,\delta)$-\DP{} w.r.t. $\varPsi$, 
i.e., for all pairs $(\lambda, \lambda')\in\varPsi$ and $R \subseteq\caly$,\, we have
$
\lift{\alg}(\lambda)[R] \leq
e^{\varepsilon}\cdot\lift{\alg}(\lambda')[R] + \delta
{.}
$
\end{definition}
\vspace{0.5ex}

Next we recall an extension~\cite{Kawamoto:19:ESORICS} of \DistP{} with a metric $\utmetric$ as follows.
Intuitively, this extended notion guarantees that when two input distributions are closer,
then the output distributions must be less distinguishable.

\vspace{0.5ex}
\begin{definition}[Extended distribution privacy]\label{def:max-XDP-dist}\rm
Let 
$d:\!(\Dists\calx\allowbreak\times\Dists\calx)\rightarrow\reals$ be a metric, and $\varPsi\subseteq\Dists\calx\times\Dists\calx$.
We say that a mechanism $\alg:\calx\rightarrow\Dists\caly$ provides \emph{$(\varepsilon,d,\delta)$-extended distribution privacy (\XDistP{}) w.r.t.} $\varPsi$ if the lifting $\lift{\alg}$ provides $(\varepsilon,d,\delta)$-\XDP{} w.r.t.~$\varPsi$,
i.e., for all $(\lambda, \lambda')\in\varPsi$ and 
$R\subseteq\caly$,
we have
$
\lift{\alg}(\lambda)[R] \leq
e^{\varepsilon d(\lambda,\lambda')}\cdot
\lift{\alg}(\lambda')[R]+ \delta
{.}
$
\end{definition}
\vspace{0.5ex}

Analogously to \XDP{}, noise should be added proportionally to the distance $d(\lambda,\lambda')$.

\begin{figure}[t]
\begin{minipage}[c]{0.5\textwidth}
\makeatletter
\def\@captype{table}
\makeatother
  \centering
  \vspace{1.5ex}
  \caption{Instances of $f$-divergence}
  \label{table:f-divergence}
  \centering
  \scalebox{0.95}[0.95]{
  \begin{tabular}{ll}
    \hline
    Divergence  & $f(t)$  \\
    \hline \hline
    KL-divergence & $t\log t$ \\
    Reverse KL-divergence\!& $-t\log t$ \\
    Total variation & $\frac{1}{2}|t-1|$ \\
    $\chi^2$-divergence & $(t-1)^2$ \\
    Hellinger distance & $\frac{1}{2}(\sqrt{t}-1)^2$\!\\[0.3ex]
    \hline
  \end{tabular}
  }
\end{minipage}
\end{figure}

\subsection{Divergence}
\label{sub:divergence}

A \emph{divergence} over a non-empty set $\caly$ is a function $\diverge{\cdot\!}{\!\cdot\!}\!:\allowbreak \Dists\caly\times\Dists\caly\rightarrow\realsnng$ such that 
for all $\mu,\mu'\in\Dists\caly$, 
(i)
$\diverge{\mu}{\mu'}\ge0$ and 
(ii)
$\diverge{\mu}{\mu'}=0$ iff $\mu=\mu'$.
Note that a divergence may not be symmetric or subadditive.
We denote by $\Div{\caly}$ the set of all divergences over $\caly$.

Next we recall the notion of (approximate) max divergence, which can be used to define \DP{}.
\begin{definition}[Max divergence]\label{def:max-divergence}\rm
Let $\delta\in[0,1]$ and $\mu, \mu'\in\Dists\caly$.
Then \emph{$\delta$-approximate max divergence} between $\mu$, $\mu'$ is:
\[
\appmaxdiverge{\mu}{\mu'} = 
\max_{\substack{R\subseteq\supp(\mu),  \mu[R]\ge\delta}} \ln \textstyle\frac{ \mu[R] - \delta }{ \mu'[R] }
{.}
\]
\end{definition}

We recall the notion of the $f$-divergences~\cite{CSISZAR:67:SSMH}.
As shown in Table~\ref{table:f-divergence},
many divergence notions (e.g. Kullback-Leibler-divergence~\cite{Kullback:51:AMS}) are instances of 
$f$-divergence.

\begin{definition}[$f$-divergence]\label{def:f-divergence}\rm
Let $\calf$ be the collection of functions defined by:
\[
\calf = \{ f: \reals^+ \rightarrow \reals^+ \mid f \mbox{ is convex and } f(1)=0 \}.
\]
Let $\caly$ be a finite set, and $\mu, \mu'\in\Dists\caly$ such that for every $y\in\caly$, $\mu'[y] = 0$ implies $\mu[y] = 0$.
Then for an $f\in\calf$, the \emph{$f$-divergence} of $\mu$ from $\mu'$ is defined as:
\[~~~~
\fdiverge{\mu}{\mu'} = 
\sum_{\substack{y\in\supp(\mu')}}
\mu'[y]\, f\bigl(
\textstyle\frac{ \mu[y]}{ \mu'[y] }
\bigr)
{.}
\]
\end{definition}

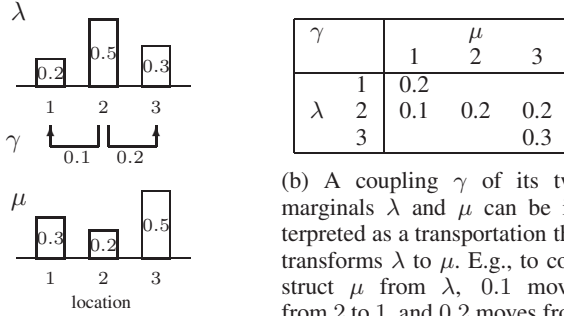
\begin{figure}[t]
\centering
\begin{subfigure}[ht]{0.24\textwidth}
\centering
\begin{picture}(65, 120)
 \put(  0,110){$\lambda$}
 \put( -2, 62){$\gamma$}
 \put(  2, 20){\line( 1, 0){66}}
 \put(  2, 85){\line( 1, 0){66}}
 \linethickness{0.6pt}
 \put( 10, 85){\framebox(10,10){\scriptsize $0.2$}}
 \put( 13, 75){\scriptsize $1$}
 \put( 30, 85){\framebox(10,25){\scriptsize $0.5$}}
 \put( 33, 75){\scriptsize $2$}
 \put( 50, 85){\framebox(10,15){\scriptsize $0.3$}}
 \put( 53, 75){\scriptsize $3$}

 \put(  0, 40){$\mu$}
 \put( 10, 20){\framebox(10,15){\scriptsize $0.3$}}
 \put( 13, 10){\scriptsize $1$}
 \put( 30, 20){\framebox(10,10){\scriptsize $0.2$}}
 \put( 33, 10){\scriptsize $2$}
 \put( 50, 20){\framebox(10,25){\scriptsize $0.5$}}
 \put( 53, 10){\scriptsize $3$}
 \put( 23,  1){\scriptsize location}

 \linethickness{0.8pt}
 \put( 20, 56){\scriptsize $0.1$}
 \put( 40, 56){\scriptsize $0.2$}

 \put( 33, 62){\line( 0, 1){8}}
 \put( 37, 62){\line( 0, 1){8}}
 \put( 15, 62){\line( 1, 0){18}}
 \put( 37, 62){\line( 1, 0){18}}
 \put( 15, 62){\vector( 0, 1){8}}
 \put( 55, 62){\vector( 0, 1){8}}
\end{picture}
\caption{An original distribution $\lambda$ and a target distribution $\mu$.}
\label{fig:eg:two-marginals}
\end{subfigure}
~\hspace{0.5ex}~
\begin{subfigure}[ht]{0.23\textwidth}
\centering
\begin{small}
\begin{tabular}{|cc|ccc|}
\hline
$\gamma$
 &   &     & $\mu$ & \\
 &   & 1   & 2   & 3 \\ \hline
 & 1 & 0.2 &     & \\
$\lambda$
 & 2 & 0.1 & 0.2 & 0.2 \\
 & 3 &     &     & 0.3 \\\hline
\end{tabular}
\end{small}
\caption{A coupling $\gamma$ of its two marginals $\lambda$ and~$\mu$ can be interpreted as a transportation that transforms $\lambda$ to $\mu$. 
E.g., to construct $\mu$ from $\lambda$,\, $0.1$ moves from $2$ to $1$, and $0.2$ moves from $2$ to $3$.}
\label{fig:eg:coupling}
\end{subfigure}
\caption{A coupling $\gamma$ that transforms $\lambda$ to $\mu$.}
\label{fig:eg:coupling-mechanism}
\end{figure}

\subsection{Probability Coupling}
\label{sub:probability-coupling}

We recall the notion of probability coupling as follows.

\begin{example}[Coupling as transformation of distributions]\label{eg:coupling-def}
Let us consider two distributions $\lambda$ and $\mu$ shown in 
Fig.~\ref{fig:eg:coupling-mechanism}.
A coupling $\gamma$ of $\lambda$ and $\mu$ shows a way of transforming $\lambda$ to $\mu$.
For example,
$\gamma[2, 1] = 0.1$ moves from $\lambda[2]$ to $\mu[1]$.
\end{example}

Formally, a coupling is defined as follows.

\begin{definition}[Coupling]\label{def:coupling}\rm
Given $\lambda\in\Dists\calx_0$ and $\mu\in\Dists\calx_1$, a \emph{coupling} of $\lambda$ and $\mu$ is a joint distribution $\gamma\in\Dists(\calx_0\times \calx_1)$ such that $\lambda$ and $\mu$ are $\gamma$'s marginal distributions, i.e.,
for each $x_0\in \calx_0$,
$\lambda[x_0] =\!\sum_{x'_1\in \calx_1}\!\gamma[x_0, x'_1]$ and
for each $x_1\in \calx_1$,
$\mu[x_1] =\!\sum_{x'_0\in \calx_0}\!\gamma[x'_0, x_1]$.
We denote by $\cp{\lambda}{\mu}$ the set of all couplings of $\lambda$ and $\mu$.
\end{definition}

\subsection{$p$-Wasserstein Metric}
\label{sub:wasserstein-metric}

Then we recall the $p$-Wasserstein metric~\cite{Vaserstein:69:PPI} 
between two distributions, which is defined using a coupling as follows.

\begin{definition}[$p$-Wasserstein metric]\label{def:p-Wasserstein-metric}\rm
Let $\utmetric$ be a metric over $\calx$, and $p\in\realsone\cup\{\infty\}$.
The \emph{$p$-Wasserstein metric} $\Wpu$ w.r.t. $\utmetric$ is defined by:
for any two distributions $\lambda, \mu\in\Dists\calx$, 
\[
\Wpu(\lambda, \mu) = 
\min_{\gamma\in \cp{\lambda}{\mu}}
\bigl(\hspace{-6ex}
\sum_{\hspace{6ex}(x_0, x_1)\in\supp(\gamma)}\hspace{-6ex}
\utmetric(x_0, x_1)^p \gamma[x_0, x_1]
\bigr)^{\!\frac{1}{p}}
{.}
\]
$\EMDu$ is also called the \emph{Earth mover's distance}.
\end{definition}

The intuitive meaning of $\EMDu(\lambda, \mu)$ is 
the minimum cost of transportation from $\lambda$ to $\mu$ in transportation theory.
As illustrated in 
Fig.~\ref{fig:eg:coupling-mechanism},
we regard the distribution $\lambda$ (resp. $\mu$) as the set of points where each point $x$ has weight $\lambda[x]$ (resp. $\mu[x]$), and we move some weight in $\lambda$ from a point $x_0$ to another $x_1$ to construct $\mu$.
We represent by $\gamma[x_0, x_1]$ the amount of weight moved from $x_0$ to $x_1$.%
\footnote{The amount of weight moved from a point $x_0$ in $\lambda$ is given by $\lambda[x_0] = \sum_{x'_1\in\calx} \gamma[x_0, x'_1]$,
while the amount moved into $x_1$ in $\mu$ is given by $\mu[x_1] = \sum_{x'_0\in\calx} \gamma[x'_0, x_1]$.
Hence $\gamma$ is a coupling of $\lambda$ and $\mu$.
}
We denote by $\utmetric(x_0, x_1)$ the cost (i.e., distance) of move from $x_0$ to $x_1$.
Then the minimum cost of the whole transportation is:
\[
\EMDu(\lambda, \mu) = 
\min_{\gamma\in \cp{\lambda}{\mu}}\hspace{-5ex}
\sum_{\hspace{6ex}(x_0, x_1)\in\supp(\gamma)}\hspace{-6ex}
\utmetric(x_0, x_1)\, \gamma[x_0, x_1]
{.}
\]
E.g., in Fig.~\ref{fig:eg:coupling-mechanism},
when the cost function $d$ is the Euclid distance over $\calx$ (e.g., $d(2, 1) = |2 - 1| = 1$), the transportation $\gamma$ achieves the minimum cost $0.1\cdot 1 + 0.2\cdot 1 = 0.3$.

Let $\GammaP$ the set of all couplings achieving $\Wpu$; i.e.,
\[
\GammaP(\lambda, \mu) = 
\argmin_{\gamma\in \cp{\lambda}{\mu}}
\bigl(\hspace{-6ex}
\sum_{\hspace{6ex}(x_0, x_1)\in\supp(\gamma)}\hspace{-6ex}
\utmetric(x_0, x_1)^p \gamma[x_0, x_1]
\bigr)^{\!\frac{1}{p}}
{.}
\]
Then $\gamma\in\GammaEMD(\lambda, \mu)$ can be efficiently computed by the \emph{North-West corner rule}~\cite{Hoffman:63:PMAMS} when $\utmetric$ is \emph{submodular} \footnote{$\utmetric$ is submodular if $\utmetric(x_0, x_1) + \utmetric(x'_0, x'_1) \leq
\utmetric(x'_0, x_1) + \utmetric(x_0, x'_1)$.}.

\section{Divergence Distribution Privacy}
\label{sec:f-distDP}
In this section we introduce new definitions of distribution privacy generalized to an arbitrary divergence $D$.
The main motivation is to discuss distribution privacy based on $f$-divergences, especially Kullback-Leibler divergence, which models ``on-average'' risk.

\subsection{Divergence \DP{} and Divergence \XDP{}}
\label{subsec:XDP-general}

To generalize distribution privacy notions, we first present a generalized formulation of point privacy parameterized with a divergence $D$.
Intuitively, we say that a randomized algorithm $\alg$ provides $(\varepsilon, D)$-\DP{} if a divergence $D$ cannot distinguish the input to $\alg$ by observing an output of $\alg$.

\begin{definition}[Divergence \DP{} w.r.t. adjacency relation]\label{def:div-DP}\rm
For an adjacency relation $\varPhi\subseteq\calx\times\calx$ and a divergence $D\in\Div{\caly}$, we say that a randomized algorithm $\alg:\calx\rightarrow\Dists\caly$ provides \emph{$(\varepsilon, D)$-\DP{} w.r.t.} $\varPhi$ if for all $(x, x')\!\in\!\varPhi$, we have
$
\diverge{\alg(x)\!}{\!\alg(x')}\!\le\!\varepsilon
\mbox{ ~and~ }
\diverge{\alg(x')\!}{\!\alg(x)}\!\le\!\varepsilon
$
{.}
\end{definition}

Note that some instances of divergence \DP{} are known in the literature.
In~\cite{Barber:14:arXiv}, 
$(\varepsilon, \Df)$-\DP{} is called \emph{$\varepsilon$-$f$-divergence privacy}, 
$(\varepsilon, \DKL)$-\DP{} (\KLP{}) is called \emph{$\varepsilon$-KL-privacy}, and
$(\varepsilon, \DTV)$-\DP{} is called \emph{$\varepsilon$-total variation privacy}.
Furthermore, $(\varepsilon,\mathit{D}_{\infty}^{\delta})$-\DP{} is equivalent to $(\varepsilon,\delta)$-\DP{}, since it is known that $(\varepsilon,\delta)$-\DP{} can be defined using the approximate max divergence~$\mathit{D}_{\infty}^{\delta}$ as follows:
\begin{restatable}{prop}{maxdivDP}
\label{prop:max-div-DP}
A randomized algorithm $\alg: \calx \rightarrow \Dists\caly$ provides $(\varepsilon,\delta)$-\DP{} w.r.t. $\varPhi\subseteq\calx\times\calx$ iff for any $(x, x')\in\varPhi$,\,
$\appmaxdiverge{\alg(x)}{\alg(x')} \le \varepsilon$ and
$\appmaxdiverge{\alg(x')}{\alg(x)} \le~\varepsilon$.
\end{restatable}

Next we generalize the notion of extended differential privacy (\XDP{})
to an arbitrary divergence $D$ as follows.

\begin{definition}[Divergence \XDP{}]\label{def:div-XDP}\rm
Let $d: \calx\times\calx\rightarrow\reals$ be a metric, $\varPhi\subseteq\calx\times\calx$, and $D\in\Div{\caly}$.
We say that a randomized algorithm $\alg:\calx\rightarrow\Dists\caly$ provides \emph{$(\varepsilon,d, D)$-\XDP{} w.r.t.} $\varPhi$ if for all $(x, x')\in\varPhi$,\,
$
\diverge{\alg(x)}{\alg(x')} \le \varepsilon d(x,x')
{.}
$
\end{definition}

These notions will be used to define (extended) divergence distribution privacy in the next section.

\subsection{Divergence \DistP{} and Divergence \XDistP{}}
\label{subsec:general:DistP}

In this section we generalize the notion of (extended) distribution privacy to an arbitrary divergence $D$.
The main aim of generalization is to present theoretical properties of distribution privacy in a more general form, and also to discuss distribution privacy based on the $f$-divergences.

Intuitively, we say that a randomized algorithm $\alg$ provides $(\varepsilon, D)$-distribution privacy w.r.t. a set $\varPsi$ of pairs of distributions if for each pair $(\lambda_0, \lambda_1)\in\varPsi$, a divergence $D$ cannot distinguish which distribution (of $\lambda_0$ and $\lambda_1$) is used to generate $\alg$'s input value.

\begin{definition}[Divergence \DistP{}]\label{def:sDistP}\rm
Let 
$D\in\Div{\caly}$, and $\varPsi\subseteq\Dists\calx\times\Dists\calx$.
We say that a randomized algorithm $\alg:\calx\rightarrow\Dists\caly$ provides \emph{$(\varepsilon,D)$-distribution privacy (\DistP{}) w.r.t.}\,$\varPsi$ if the lifting $\lift{\alg}$ provides $(\varepsilon,D)$-\DP{} w.r.t. $\varPsi$,
i.e., for all $(\lambda, \lambda')\in\varPsi$,\,
\vspace{-2ex}
\begin{align*}
\diverge{\lift{\alg}(\lambda)}{\lift{\alg}(\lambda')} \le \varepsilon
{.}
\end{align*}
\end{definition}
\vspace{0.5ex}

As with the generalization of \DP{} to $f$-divergence~\cite{Barber:14:arXiv},\, $D$-\DistP{} expresses privacy against an adversary performing the hypothesis test corresponding to the divergence~$D$.
When $D$ involves averaging (e.g., $D=\DKL$), $D$-\DistP{} formalizes ``on-average'' privacy, which relaxes the original \DistP{}.

Next we introduce \XDistP{} parameterized with a divergence $D$.
Intuitively, \XDistP{} with a divergence $D$ guarantees that when two input distributions $\lambda$ and $\lambda'$ are closer (in terms of a metric $d$), then the output distributions $\lift{\alg}(\lambda)$ and $\lift{\alg}(\lambda')$ must be less distinguishable (in terms of $D$).

\begin{definition}[Divergence \XDistP{}]\label{def:sXDP-dist}\rm
Let 
$d$ be a metric over $\Dists\calx$, 
$D\in\Div{\caly}$, and $\varPsi\subseteq\Dists\calx\times\Dists\calx$.
We say that a randomized algorithm $\alg:\calx\rightarrow\Dists\caly$ 
provides \emph{$(\varepsilon,d, D)$-extended distribution privacy (\XDistP{}) w.r.t.} $\varPsi$ if the lifting $\lift{\alg}$ provides $(\varepsilon,d, D)$-\XDP{} w.r.t. $\varPsi$,
i.e., for all $(\lambda, \lambda')\in\varPsi$,\,
\vspace{-3ex}
\begin{align*}
\diverge{\lift{\alg}(\lambda)}{\lift{\alg}(\lambda')} \le \varepsilon d(\lambda,\lambda')
{.}
\end{align*}
\end{definition}
\vspace{1ex}

\section{Properties of Divergence Distribution Privacy}
\label{sec:properties}
In this section we show useful properties of divergence distribution privacy, such as compositionality and relationships among distribution privacy notions.

\subsection{Basic Properties of Divergence Distribution Privacy}
\label{sub:proofs:summary:properties}

In Tables~\ref{table:summary:basic-properties:DistP:f} and~\ref{table:summary:basic-properties:XDistP:f} we summarize the results on two kinds of sequential compositions $\Seq$ (Fig.~\ref{fig:sequential-shared}) and~$\liftSeq$ (Fig.~\ref{fig:sequential-independent}), post-processing, and pre-processing for divergence \DistP{} and for divergence \XDistP{}, respectively.
We present the details and proofs for these results in
\conference{\cite{Kawamoto:Allerton19:arXiv}.}
\arxiv{Appendices~\ref{sub:composition:Seq:details}, \ref{sub:composition:liftSeq:details}, and~\ref{sub:post-pre-processing}.}

The two kinds of composition have been studies in previous work (e.g.,~\cite{Kawamoto:17:LMCS,Kawamoto:19:ESORICS}).
For two mechanisms $\alg_0$ and $\alg_1$, 
the composition $\alg_1 \Seq \alg_0$ means that an identical input value $x$ is given to two \DistP{} mechanisms $\alg_0$ and $\alg_1$,
whereas$\alg_1 \liftSeq \alg_0$ means that independent inputs $x_b$ are provided to mechanisms $\alg_b$.
Note that this kind of composition is adaptive in the sense that the output of $\alg_1$ can be dependent on that of $\alg_0$.
Hence the compositonality does not hold in general for $f$-divergence, 
whereas we show the compositionality for KL-divergence in Tables~\ref{table:summary:basic-properties:DistP:f} and~\ref{table:summary:basic-properties:XDistP:f}.
For non-adaptive sequential composition, the compositionality of divergence \DistP{}/\XDistP{} is straightforward from~\cite{Barthe:13:ICALP}, which show the compositionality of popular $f$-divergences, including total variation
and Hellinger distance.

As for pre-processing, we use the following definition of stability~\cite{Kawamoto:19:ESORICS}, which is analogous to the stability for \DP{}.
\begin{definition}[Stability]\label{def:c-stable}\rm
Let $c\in\nats$, 
$\varPsi\subseteq\Dists\calx\times\Dists\calx$,
and $\sensfunc$ be a metric over $\Dists\calx$.
A transformation $T:\Dists\calx\rightarrow\Dists\calx$ is \emph{$(c, \varPsi)$-stable} if 
for any $(\lambda_0,\lambda_1)\in\varPsi$,
$T(\lambda_0)$ can be reached from $T(\lambda_1)$ at most $c$-steps over $\varPsi$.\,
Analogously, $T:\Dists\calx\rightarrow\Dists\calx$ is \emph{$(c,\sensfunc)$-stable} if for any $\lambda_0,\lambda_1\in\Dists\calx$, 
$\sensfunc(T(\lambda_0),T(\lambda_1)) \le c \sensfunc(\lambda_0,\lambda_1)$.
\end{definition}

\begin{figure}[t]\label{fig:compositions}
\centering
\begin{subfigure}[t]{0.21\textwidth}
\centering
\begin{picture}(110, 83)
 \put( 52, 69){$\alg_1 \Seq \alg_0$}
 \thicklines
 \put( 48, 40){\framebox(40,20){$\alg_0$}}
 \put( 48,  0){\framebox(40,20){$\alg_1$}}

 \put(  10,  30){\vector(  1,  0){20}}
 \put(  31,  10){\line(  0,  1){40}}
 \put(  31,  50){\vector(  1,  0){16}}
 \put(  31,  10){\vector(  1,  0){16}}
 \put(  68,  38){\vector(  0, -1){16}}
 \put(  90,  10){\vector(  1,  0){28}}
 \thicklines
 \put(  16,  36){$x$}
 \put(   0,  27){$\lambda$}
 \put(  73,  29){$y_0$}
 \put( 100,  16){$y_1$}
 \put( 42, -3){\dashbox{1.0}(51,67){}}
\end{picture}
\caption{$\Seq$ with shared input.\label{fig:sequential-shared}}
\end{subfigure}\hspace{1ex}\hfill
\begin{subfigure}[t]{0.22\textwidth}
\centering
\begin{picture}(105, 75)
 \put( 41, 69){$\alg_1 \liftSeq \alg_0$}
 \thicklines
 \put( 38, 40){\framebox(40,20){$\alg_0$}}
 \put( 38,  0){\framebox(40,20){$\alg_1$}}

 \put( 11,  50){\vector(  1,  0){25}}
 \put( 11,  10){\vector(  1,  0){25}}
 \put( 58,  38){\vector(  0, -1){16}}
 \put( 80,  10){\vector(  1,  0){28}}
 \thicklines
 \put( 15,  56){$x_0$}
 \put( 15,  16){$x_1$}
 \put(  0,  48){$\lambda$}
 \put(  0,   8){$\lambda$}
 \put(  63,  29){$y_0$}
 \put(  90,  16){$y_1$}
 \put( 32, -3){\dashbox{1.0}(51,67){}}
\end{picture}
\caption{$\liftSeq$ with independent inputs.\label{fig:sequential-independent}}
\end{subfigure}\hspace{0.4ex}\hfill
\caption{Two kinds of sequential compositions $\Seq$ and $\liftSeq$.\label{fig:sequentials}}
\end{figure}

\begin{table*}[ht]
  \centering
  \vspace{3.5ex}
  \caption{Summary of basic properties of divergence \DistP{}.}
  \label{table:summary:basic-properties:DistP:f}
  \scalebox{0.96}[0.96]{
  \footnotesize
  \renewcommand{\arraystretch}{1.2}
  \begin{tabular}{lll}
    \hline
    Sequential composition $\Seq$ ($\DKL$) &
      $\alg_b$ is $(\varepsilon_b,\DKL)$-\DistP{} \\
    &
      $\Rightarrow$
      $\alg_1 \Seq \alg_0$ is $(\varepsilon_0+\varepsilon_1, \DKL)$-\DistP{} \\
    \hline 
    Sequential composition $\liftSeq$ ($\DKL$) &
      $\alg_b$ is $(\varepsilon_b,\DKL)$-\DistP{} \\
    &
      $\Rightarrow$
      $\alg_1 \liftSeq \alg_0$ is $(\varepsilon_0+\varepsilon_1, \DKL)$-\DistP{} \\
    \hline 
    Post-processing &
      $\alg_0$ is $(\varepsilon,\Df)$-\DistP{}
      $\Rightarrow$
      $\alg_1\circ\alg_0$ is $(\varepsilon,\Df)$-\DistP{} \\
    \hline
    Pre-processing (by $c$-stable $T$) &
      $\alg$ is $(\varepsilon,D)$-\DistP{}
      $\Rightarrow$
      $\alg\circ T$ is $(c\,\varepsilon,D)$-\DistP{} \\
    \hline
  \end{tabular}
  }\vspace{1ex}
\end{table*}

\begin{table*}[ht]
  \centering
  \caption{Summary of basic properties of divergence \XDistP{}.}
  \label{table:summary:basic-properties:XDistP:f}
  \scalebox{0.96}[0.96]{
  \footnotesize
  \renewcommand{\arraystretch}{1.2}
  \begin{tabular}{lll}
    \hline
    Sequential composition $\Seq$ ($\DKL$) &
      $\alg_b$ is $(\varepsilon_b,\sensemd,\DKL)$-\XDistP{} \\
    &
      $\Rightarrow$
      $\alg_1 \Seq \alg_0$ is $(\varepsilon_0+\varepsilon_1, \sensemd, \DKL)$-\XDistP{} \\
    \hline 
    Sequential composition $\liftSeq$ ($\DKL$) &
      $\alg_b$ is $(\varepsilon_b,\sensemd,\DKL)$-\XDistP{} \\
    &
      $\Rightarrow$
      $\alg_1 \liftSeq \alg_0$ is $(\varepsilon_0+\varepsilon_1, \sensemd, \DKL)$-\XDistP{} \\
    \hline 
    Post-processing &
      $\alg_0$ is $(\varepsilon,\sensfunc,\Df)$-\XDistP{}
      \!$\Rightarrow$\!
      $\alg_1{\circ}\alg_0$ is $(\varepsilon,\sensfunc,\Df)$-\XDistP{} \\
    \hline
    Pre-processing (by $c$-stable $T$) &
      $\alg$ is $(\varepsilon,\sensfunc,D)$-\XDistP{}
      \!$\Rightarrow$\!
      $\alg\circ T$ is $(c\,\varepsilon,\sensfunc,D)$-\XDistP{} \\
    \hline
  \end{tabular}
  }
\end{table*}

\subsection{Relationships among Distribution Privacy Notions}
\label{sub:basic-properties:gDistP}

In Fig.~\ref{fig:relation-properties} we show the summary of the relationships among notions of divergence \XDP{} and divergence \XDistP{}.
\arxiv{
See Appendices~%
\ref{sub:proofs:DistP-implies-DP} 
and~\ref{subsec:relationships:notions}}
\conference{
See Appendix~\ref{sub:proofs:DistP-implies-DP}
and~\cite{Kawamoto:Allerton19:arXiv}}
for details and proofs.

\begin{figure*}[t]
\centering
\begin{picture}(405, 85)
 \put( 49, 73){\scriptsize Theorem~2 of~\cite{Kawamoto:19:ESORICS}}
 \put( 65, 45){\tiny if $\delta=0$}
 \put( 55, 37){\scriptsize Theorem~\ref{thm:DPfromDistP}}
 \arxiv{\put(183, 73){\scriptsize Proposition~\ref{prop:relationDDistP}}}
 \conference{\put(173, 73){\scriptsize Proposition 10 in~\cite{Kawamoto:Allerton19:arXiv}}}
 \put(190, 60){\tiny if $\delta=0$}
 \arxiv{\put(183, 23){\scriptsize Proposition~\ref{prop:relationDDistP}}}
 \conference{\put(173, 25){\scriptsize Proposition 10 in~\cite{Kawamoto:Allerton19:arXiv}}}
 \put(190, 14){\tiny if $\delta=0$}
 \arxiv{\put( 82, 25){\scriptsize Proposition~\ref{prop:relationWDistP}}}
 \conference{\put( 70, 25){\scriptsize Proposition 9 in~\cite{Kawamoto:Allerton19:arXiv}}}
 \arxiv{\put(270, 25){\scriptsize Proposition~\ref{prop:relationWDistP}}}
 \conference{\put(270, 25){\scriptsize Proposition 9 in~\cite{Kawamoto:Allerton19:arXiv}}}
 \put(313, 73){\scriptsize Theorem~\ref{thm:f-XDP-dist}}
 \put(319, 45){\tiny if $\delta=0$}
 \put(313, 37){\scriptsize Theorem~\ref{thm:DPfromDistP}}
 \linethickness{0.6pt}
 \put( -5, 47){\framebox(60,20){\scriptsize $(\varepsilon, \utmetric, \Dinf^{\delta})$-\XDP}}
 \put( 93, 47){\framebox(89,20){\scriptsize $(\varepsilon, \sensinf, \Dinf^{\delta\cdot|\varPhi|}\!)$-\XDistP}}
 \put( 93,  0){\framebox(89,20){\scriptsize $(\varepsilon, \sensemd, \Dinf^{\delta\cdot|\varPhi|}\!)$-\XDistP}}
 \put(224, 47){\framebox(89,20){\scriptsize $(\varepsilon, \sensinf, \DKL)$-\XDistP}}
 \put(224,  0){\framebox(89,20){\scriptsize $(\varepsilon, \sensemd, \DKL)$-\XDistP}}
 \put(350, 47){\framebox(60,20){\scriptsize $(\varepsilon, \utmetric, \DKL)$-\XDP}}

 \linethickness{0.8pt}
 \put( 58, 61){\vector( 1, 0){31}}
 \put( 89, 53){\vector(-1, 0){31}}
 \put(184, 57){\vector( 1, 0){36}}
 \put(184, 10){\vector( 1, 0){36}}
 \put(347, 61){\vector(-1, 0){31}}
 \put(316, 53){\vector( 1, 0){31}}
 \put(135, 22){\vector( 0, 1){22}}
 \put(267, 22){\vector( 0, 1){22}}
\end{picture}
\caption{Relationships among divergence \XDistP{} notions.}
\label{fig:relation-properties}
\end{figure*}
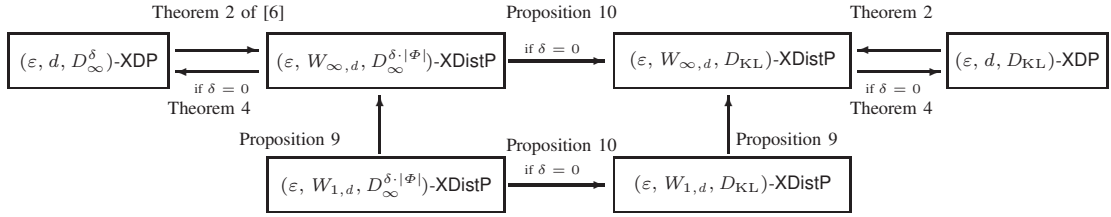

\section{Local Mechanisms for Divergence Distribution Privacy}
\label{sec:DistP-by-DP-f}
In this section we present how much degree of divergence \DistP{}/\XDistP{} can be achieved by local obfuscation.
Specifically, we show how $f$-divergence privacy contribute to the obfuscation of probability distributions.
To prove those results, we use the notion of probability coupling.

\subsection{Divergence \DistP{} by Local Obfuscation}
\label{sub:f-DistP}

We first show that $f$-divergence privacy mechanisms provide $\Df$-\DistP{}.
To present this formally, we recall the notion of the lifting of relations as follows.
\begin{definition}[Lifting of relations]\label{def:lifting-relations}\rm
Given a relation $\varPhi\subseteq\calx\times\calx$, the \emph{lifting} of $\varPhi$ is the maximum relation $\lift{\varPhi}\subseteq \Dists\calx\times\Dists\calx$ such that for any $(\lambda_0, \lambda_1)\in\lift{\varPhi}$, there exists a coupling $\gamma\in\cp{\lambda_0}{\lambda_1}$ satisfying $\supp(\gamma)\subseteq\varPhi$.
\end{definition}

Intuitively, when $\lambda_0$ and $\lambda_1$ are adjacent w.r.t. the lifted relation $\lift{\varPhi}$, then we can construct $\lambda_1$ from $\lambda_0$ according to the coupling $\gamma$, that is, only by moving mass from $\lambda_0[x_0]$ to $\lambda_1[x_1]$ where $(x_0, x_1)\in\varPhi$ (i.e., $x_0$ is adjacent to $x_1$).
Note that by Definition~\ref{def:coupling}, the coupling $\gamma$ is a probability distribution over $\varPhi$ whose marginal distributions are $\lambda_0$ and~$\lambda_1$.
If $\varPhi = \calx\times\calx$, then $\lift{\varPhi} = \Dists\calx\times\Dists\calx$.

Now we show that every $f$-divergence privacy mechanism provides $\Df$-\DistP{} as follows.
(See Appendix~\ref{sub:proofs:Df-DistP} for the proof.)

\begin{restatable}[$(\varepsilon, \Df)$-\DP{} $\Rightarrow$ $(\varepsilon, \Df)$-\DistP]{thm}{fDistP}
\label{thm:fDistP}
Let $\varPhi \subseteq \calx\times\calx$.
If a randomized algorithm $\alg:\calx\rightarrow\Dists\caly$ provides $(\varepsilon, \Df)$-\DP{} w.r.t. $\varPhi$, then
it provides $(\varepsilon, \Df)$-\DistP{} w.r.t. $\lift{\varPhi}$.
\end{restatable}

Intuitively, the $f$-divergence privacy mechanism $\alg$ makes any pair $(\lambda_0, \lambda_1)$ of input distributions in $\lift{\varPhi}$ indistinguishable in terms of $\Df$ up to the threshold~$\varepsilon$.

\subsection{Divergence \XDistP{} by Local Obfuscation}
\label{sub:f-XDistP}
Next we investigate how much noise should be added for local obfuscation mechanisms to provide divergence \XDistP{}.

We first consider two point distributions $\lambda_0$ at $x_0$ and $\lambda_1$ at $x_1$, i.e., $\lambda_0[x_0] = \lambda_1[x_1] = 1$.
Then an $(\varepsilon, \utmetric, \Df)$-\XDP{} mechanism $\alg$ satisfies:
\begin{align*}
\fdiverge{\!\lift{\alg}\!(\lambda_0)\!}{\!\lift{\alg}\!(\lambda_1)\!}
=
\fdiverge{\!\alg(x_0)\!}{\!\alg(x_1)\!}
\le \varepsilon \utmetric(x_0,x_1)
{.}
\end{align*}
Hence the noise added by $\alg$ should be proportional to the distance $\utmetric(x_0, x_1)$ between $x_0$ and $x_1$.

To generalize this observation on point distributions to arbitrary distributions, we need to employ some metric between distributions.
As the metric,
we could use 
the \emph{diameter} over the supports, which is defined by:
\[
\diam(\lambda_0, \lambda_1) = 
\max_{x_0\in\supp(\lambda_0), x_1\in\supp(\lambda_1)} \utmetric(x_0, x_1)
{,}
\]
or the $\infty$-Wasserstein metric $\sensinf$, which is used for \XDistP{}~\cite{Kawamoto:19:ESORICS}.
However, when there is an outlier in $\lambda_0$ or $\lambda_1$, then
$\diam(\lambda_0, \lambda_1)$ and 
$\sensinf(\lambda_0, \lambda_1)$ tend to be large.
Since the mechanism needs to add noise proportionally to the distance $\diam(\lambda_0, \lambda_1)$ or $\sensinf(\lambda_0, \lambda_1)$ to achieve $\XDistP{}$, it needs to add large amount of noise and thus loses utility significantly.

To have better utility, we employ the Earth mover's distance ($1$-Wasserstein metric) $\EMDu$ as a metric for $\Df$-\XDistP{} mechanisms.
Given two distributions $\lambda_0$ and $\lambda_1$ over $\calx$, we consider a transportation $\gamma$ from $\lambda_0$ to $\lambda_1$ that minimizes the expected cost of the transportation.
Then the minimum of the expected cost is given by the Earth mover's distance $\EMDu(\lambda_0,\lambda_1)$.

Now we show that, to achieve $\Df$-\XDistP{}, we only have to add noise proportionally to the Earth mover's distance $\sensemd$ between the input distributions.
To formalize this, we define a lifted relation $\liftp{\varPhi}$ as the maximum relation over $\Dists\calx$ s.t. for any $(\lambda_0, \lambda_1)\in\liftp{\varPhi}$, there is a coupling $\gamma\in\cp{\lambda_0}{\lambda_1}$ satisfying $\supp(\gamma)\subseteq\varPhi$ and $\gamma\in\GammaP(\lambda_0, \lambda_1)$.

\begin{restatable}[$(\varepsilon, \utmetric, \Df)$-\XDP{} $\Rightarrow$ $(\varepsilon, \sensemd, \Df)$-\XDistP]{thm}{fXDPdist}
\label{thm:f-XDP-dist}
Let 
$\utmetric: \calx\times\calx\rightarrow\reals$ be a metric.
If a randomized algorithm $\alg:\calx\rightarrow\Dists\caly$ provides $(\varepsilon, \utmetric, \Df)$-\XDP{} w.r.t. $\varPhi$ then
it provides $(\varepsilon, \sensemd, \Df)$-\XDistP{} 
w.r.t. $\liftemd{\varPhi}$.
\end{restatable}

See Appendix~\ref{sub:proofs:Df-DistP} for the proof.
Since the Earth mover's distance is not grater than the diameter or $\infty$-Wasserstein distance, $\Df$-\XDistP{} may require less noise than $\Dinf$-\XDistP{}.

\section{Local Distribution Obfuscation with Auxiliary Inputs}
\label{sec:DistP-auxiliary}
In this section we introduce a local obfuscation mechanism which we call a \emph{coupling mechanism} in order to provide distribution privacy while optimizing utility.
Specifically, a coupling mechanism uses (full or approximate) knowledge on the input probability distributions to perturb each single input value so that the output distribution gets indistinguishable from some target probability distribution.
To define the mechanism, we calculate the probability coupling of each input distribution and the target distribution.

\subsection{Privacy Definitions with Auxiliary Inputs}

We first extend the definition of divergence \DistP{} so that a local obfuscation mechanism $\alg$ can receive some \emph{auxiliary input} (e.g. context information) ranging over a set~$\cals$, which might be used for $\alg$ to apply different randomized algorithms in different situations or to different input distributions.

\begin{definition}[Divergence \DistP{} with auxiliary inputs]\label{def:divDistP}\rm
Let $\varepsilon\in\realsnng$, $D\in\Div{\caly}$, and
$\varPsi\subseteq(\cals\times\Dists\calx)\times(\cals\times\Dists\calx)$.
We say that a randomized algorithm $\alg:\cals\times\calx\rightarrow\Dists\caly$ 
provides \emph{$(\varepsilon,D)$-distribution privacy w.r.t.}\,$\varPsi$ if for all pairs $((s, \lambda), (s', \lambda'))\in\varPsi$,\,
\vspace{-1.5ex}
\begin{align*}
\diverge{\lift{\alg}(s, \lambda)}{\lift{\alg}(s', \lambda')} \le \varepsilon
{.}
\end{align*}
\end{definition}
\vspace{0.5ex}

In this definition, the auxiliary input over $\cals$ typically represents contextual information about where the obfuscation mechanism $\alg$ is used or what distribution an input is sampled from.
Such information may be useful to customize $\alg$ to improve utility while providing distribution privacy in specific situations.
For example, assume that each auxiliary input $s$ represents the fact that an input $x$ is sampled from a distribution $\lambda_s$.
If a local mechanism $\alg$ uses this auxiliary information to always produce a distribution $\mu$ of outputs\footnote{If $\alg$ can use no auxiliary information but wants to produce $\mu$, then the output value needs to be independent of the input, hence very poor utility.},
it can prevent the leakage of information on the input distribution $\lambda_s$.
We elaborate on this in the next sections.

\subsection{Coupling Mechanisms}
\label{sub:coupling-mechanism}

In this section we introduce a new local obfuscation mechanism, which we call a \emph{coupling mechanism}.
The aim of the new mechanism is to improve the utility while protecting distribution privacy when we know the input distribution fully or approximately.
Intuitively, a coupling mechanism uses (full or partial) information on the input distribution $\lambda\in\Dists\calx$ and produces an output value following some identical distribution $\mu\in\Dists\caly$, which we call a \emph{target distribution}.
More specifically, given some auxiliary information $s$ about~$\lambda$, a coupling mechanism $\alg: \cals\times\calx\rightarrow\Dists\caly$ probabilistically maps each  input value $x$ to some output value $y$ so that $y$ is distributed over the target distribution $\mu$.

The simplest construction of a coupling mechanism would be to randomly sample a value $y$ from $\mu$ independently of the input $x$.
However, this mechanism provides very poor utility, since the output $y$ loses all information on $x$.

Instead, we construct a mechanism by calculating a coupling $\gamma\in\Dists(\calx\times\caly)$ that transforms $\lambda$ to $\mu$ with the minimum loss.
We explain this using a simple example below.

\begin{example}[Coupling mechanism]
\label{eg:coupling-mechanism}
A coupling $\gamma$ of two distributions $\lambda$ and $\mu$ (Fig.~\ref{fig:eg:coupling}) shows a way of transforming $\lambda$ to $\mu$ by probabilistically adding noise to each single input value drawn from $\lambda$.
More specifically, $\gamma[2, 1] = 0.1$ means that $0.1$ (out of $\lambda[2]=0.5$) moves from $2$ to $1$,
and $\gamma[2, 3] = 0.2$ means that $0.2$ moves from $2$ to $3$.
Based on this coupling $\gamma$, we construct the coupling mechanism $\CP$ that maps the input $2$ to the output $1$ with probability $20\% (= 0.1/0.5)$, and to the output $3$ with probability $40\% (= 0.2/0.5)$.
By applying this mechanism $\CP$ to the input distribution $\lambda$, the resulting output distribution $\lift{\CP}(\lambda)$ is identical to~$\mu$.
\end{example}

Formally, we assume that for each auxiliary input $s \in\cals$, we learn that the input distribution is approximately $\hlambda_s \in\Dists\calx$ while the actual distribution is $\lambda_s \in\Dists\calx$.
Then we define the coupling mechanism $\CP$ as follows.

\begin{definition}[Coupling mechanism]\label{def:cp-mechanism}\rm
Let $\mu\in\Dists\caly$.
For each $s\in\cals$, let $\hlambda_s \in\Dists\calx$ be an approximate input distribution, and
$\gamma_s\in\cp{\hlambda_s}{\mu}$ be a coupling of $\hlambda_s$ and $\mu$.
Then a \emph{coupling mechanism} w.r.t. $\mu$ is defined as a randomized algorithm $\CP:\cals\times\calx\rightarrow\Dists\caly$ such that given $s\in\cals$ and $x\in\calx$, outputs $y\in\caly$ with the probability:
\begin{align*}
\CP(s, x)[y] =& 
\textstyle\frac{\gamma_s[x, y]}{\hlambda_s[x]}
{.}
\end{align*}
\end{definition}
\vspace{1ex}

When $\CP$ can access the exact information on $\lambda_s$ (i.e., $\hlambda_s$ is identical to the actual distribution $\lambda_s$ from which inputs are sampled), then $\CP$ provides $(0, D)$-\DistP{} for any divergence $D$, i.e., no information on the input distribution is leaked by the output of $\CP$.
However, we often obtain only approximate information on the input distribution.
In this case, $\CP$ still provides strong privacy as shown in the next section.

\subsection{Distribution Privacy of Coupling Mechanisms}
\label{sub:coupling:DistP-Loss}

In this section we evaluate the \DistP{} and utility of coupling mechanisms.
\arxiv{(See Appendix~\ref{sub:proof:cp} for the proof.)}
\conference{(See~\cite{Kawamoto:Allerton19:arXiv} for the proof.)}

\begin{restatable}[\DistP{} of the coupling mechanism]{thm}{firstCouplingMaxDistP}
\label{thm:firstCouplingMaxDistP}
Let $\varPsi \subseteq (\cals\times\Dists\calx)\times(\cals\times\Dists\calx)$ such that each element of $\varPsi$ is of the form $(s,\lambda_s)$ for some $s\in\cals$.
Let $\CP$ be a coupling mechanism w.r.t. a target distribution $\mu$.
Assume that for each $s\in\cals$, the approximate knowledge $\hlambda_s$ is close to the actual distribution $\lambda_s$ in the sense that $\maxdiverge{\hlambda_s}{\lambda_s} \le \varepsilon$ and $\maxdiverge{\lambda_s}{\hlambda_s} \le \varepsilon$.
Then $\CP$ provides:
\begin{enumerate}
\item $(2\varepsilon, \Dinf)$-\DistP{} w.r.t.~$\varPsi$;
\item $(2\varepsilon\, e^\varepsilon, \DKL)$-\DistP{} w.r.t.~$\varPsi$;
\item $(e^\varepsilon f(e^{2\varepsilon}), \Df)$-\DistP{} w.r.t.~$\varPsi$.
\end{enumerate}
\end{restatable}

This theorem implies that when the mechanism $\CP$ learns the exact distribution, i.e., $\hlambda_s = \lambda_s$, then by $\varepsilon = 0$ it provides $(0, \Dinf)$-\DistP{}, hence there is no leaked information on the input distributions.
For $\varepsilon \approx 0$, we have $\varepsilon\, e^\varepsilon \approx \varepsilon(1 + \varepsilon) \approx \varepsilon$,
hence $\CP$ provides approximately $(2\varepsilon, \DKL)$-\DistP{} .

\subsection{Utility-Optimal Coupling Mechanisms}
\label{sub:coupling:Loss}

In this section we introduce a utility-optimal coupling mechanism.
Here we assume that there is some metric $\utmetric$ over $\calx\cup\caly$.
Then the notion of utility loss of a local obfuscation mechanism is defined as follows.
\vspace{0.5ex}
\begin{definition}[Expected utility loss]\label{def:expected-utility-loss}\rm
Given an input distribution $\lambda\in\Dists\calx$ and a metric $\utmetric$ over $\calx\cup\caly$, 
the \emph{expected utility loss} of a randomized algorithm $\alg: \calx\rightarrow\Dists\caly$ is:
\begin{align*}
\sum_{x\in\calx, y\in\caly}
\lambda[x] \alg(x)[y] \utmetric(x, y)
{.}
\end{align*}
\end{definition}
\vspace{0.5ex}

The utility loss of a coupling mechanism depends on the choice of the coupling used in the mechanism.
Given an Euclid distance $\utmetric$ and an input distribution $\hlambda_s$, the expected utility loss of a coupling mechanism w.r.t. a target distribution $\mu$ using a coupling $\gamma_s$ is represented by 
$\sum_{(x_0, x_1)\in\supp(\gamma_s)} \utmetric(x_0, x_1) \gamma_s[x_0, x_1]$.

Now we define the coupling mechanism that minimizes the expected utility loss as follows.

\begin{definition}[Utility-optimal coupling mechanism]\label{def:opt-cp-mechanism}\rm
Let $\mu\in\Dists\caly$.
A \emph{utility-optimal coupling mechanism} w.r.t. $\mu$ is a coupling mechanism w.r.t. $\mu$ that uses a coupling $\gamma_s\in\GammaEMD(\lambda_{s}, \mu)$ for each $s\in\cals$.
\end{definition}

\begin{restatable}[Loss of the coupling mechanism]{prop}{CouplingLoss}
\label{prop:CouplingLoss}
For each $s\in\cals$, the expected utility loss of a utility-optimal coupling mechanism w.r.t. a target distribution $\mu\in\Dists\caly$ is given by the Earth mover's distance $\EMDu(\hlambda_s, \mu)$.
\end{restatable}

The proof is straightforward from the definition of the Earth mover's distance.
Note that as mentioned in Section~\ref{sub:wasserstein-metric}, the coupling $\gamma_s\in\GammaEMD(\lambda_{s}, \mu)$ can be efficiently calculated by the North-West corner rule when $\utmetric$ is submodular.

Analogously, we could define a coupling mechanism that minimizes the maximum loss by using a coupling $\gamma_s\in\GammaInf(\lambda_{s}, \mu)$ for each $s\in\cals$.
Then the worst-case utility loss is given by the $\infty$-Wasserstein metric $\Winfu(\hlambda_s, \mu)$.

\section{Related work}
\label{sec:related-work}
Since the seminal work of Dwork~\cite{Dwork:06:ICALP} on 
differential privacy (\DP{}), 
a lot of its variants have been studied 
to provide different types of privacy guarantees~\cite{Desfontaines:19:arXiv}; 
e.g., 
$d$-privacy~\cite{Chatzikokolakis:13:PETS}, 
$f$-divergence privacy~\cite{Barthe:13:ICALP,Barber:14:arXiv}, 
mutual-information \DP{}~\cite{Cuff:16:CCS}, 
concentrated \DP{}~\cite{Dwork:16:arXiv},
R\'enyi \DP{}~\cite{Mironov:17:CSF}, 
Pufferfish privacy~\cite{Kifer:12:PODS}, 
Bayesian \DP{}~\cite{Yang:15:SIGMOD}, 
local \DP{}~\cite{Duchi:13:FOCS}, 
personalized \DP{}~\cite{Jorgensen:15:ICDE},
and
utility-optimized local \DP{}~\cite{Murakami:19:USENIX}.
All of these are intended to 
protect single input values 
instead of input distributions.

A few researches have explored the privacy of distributions.
Jelasity {\it et al}.~\cite{Jelasity:IHMMSec:14} propose distributional \DP{} to protect the privacy of distribution parameters $\theta$ in a Bayesian style (unlike \DP{} and \DistP{}).
Kawamoto {\it et al}.~\cite{Kawamoto:19:ESORICS} propose the \DistP{} notion in a \DP{} style.
Geumlek {\it et al}.~\cite{Geumlek:19:ISIT} propose profile-based privacy, a variant of \DistP{} that allows the mechanisms to depend on the perfect knowledge of input distributions.
However, these studies deal only with the worst-case risk, and neither relax them to the average-case risk (with divergence) nor allow them to use arbitrary auxiliary information (in spite that available information on input distributions is often approximate only).

There have been many studies (e.g.,~\cite{Xu:13:VLDB}) on the \emph{\DP{} of histogram publishing}, which is different from \DistP{} as follows.
Histogram publishing is a \emph{central} mechanism that \emph{hides a single record} $x\in\calx$ and outputs an obfuscated histogram, e.g., $\mu\in\Dists\caly$, 
whereas a \DistP{} mechanism is a \emph{local} mechanism that aims at \emph{hiding an input distribution} $\lambda\in\Dists\calx$ and outputs a single perturbed value $y\in\caly$.
As explained in~\cite{Kawamoto:19:ESORICS}, neither of these implies the other.

\section{Conclusion}
\label{sec:conclusion}
We introduced the notions of divergence \DistP{} and presented their useful theoretical properties in a general form.
By using probability coupling techniques, we presented how much divergence \DistP{} can be achieved by local obfuscation.
In particular, we proved that the perturbation noise should be added proportionally to the Earth mover's distance between the input distributions.
We also proposed a local mechanism called a (utility-optimal) coupling mechanism and theoretically evaluated their \DistP{} and utility loss in the presence of (exact or approximate) knowledge on the input distributions.

As for future work, we are planning to develop various kinds of coupling mechanisms for specific applications, such as location privacy.

\bibliographystyle{IEEEtran}
\bibliography{short}

\appendix

\subsection{Local Mechanisms for $\Df$-\DistP{}/\XDistP{}}
\label{sub:proofs:Df-DistP}

We first show the proofs for the $\Df$-\DistP{}/\XDistP{} achieved by local obfuscation mechanisms.
\vspace{1ex}

\fDistP*

\begin{proof}
Let $(\lambda_0, \lambda_1)\in\lift{\varPhi}$
and $\Gamma \eqdef \cp{\lambda_0}{\lambda_1}$.
\begin{align}
&
\fdiverge{\lift{\alg}(\lambda_0)}{\lift{\alg}(\lambda_1)} 
\nonumber
\\ =&
{\textstyle \sum_{y}}\,
\lift{\alg}(\lambda_1)[y]~
f\bigl({\textstyle \frac{ \lift{\alg}(\lambda_0)[y]}{ \lift{\alg}(\lambda_1)[y] } } \bigr)
\nonumber
\\ =&
{\textstyle \sum_{y}}\, 
{\textstyle \sum_{x_1}}\,
\lambda_1[x_1]\,\alg(x_1)[y]~
f\Bigl( {\textstyle \frac{ \sum_{x_0}\lambda_0[x_0]\,\alg(x_0)[y] }{ \sum_{x_1}\lambda_1[x_1]\,\alg(x_1)[y] } } \Bigr)
\nonumber
\\ =&
\min_{\gamma\in\Gamma}
{\textstyle \sum_{y,x_0,x_1}}\,
\gamma[x_0,x_1]\,\alg(x_1)[y]~
f\Bigl( {\textstyle\frac{ \sum_{x_0,x_1}\hspace{-0.5ex}\gamma[x_0,x_1]\,\alg(x_0)[y] }{\!\sum_{x_0,x_1}\hspace{-0.5ex}\gamma[x_0,x_1]\,\alg(x_1)[y] }} \Bigr)
\nonumber
\\[-0.4ex]
&\hspace{19ex}\text{(where $(x_0, x_1)$ ranges over $\supp(\gamma)$)}
\nonumber
\\ =&
\min_{\gamma\in\Gamma}
{\textstyle \sum_{y}}\,c~
f\Bigl( {\textstyle\frac{1}{c}} \,{\textstyle \sum_{x_0,x_1}}\,
\gamma[x_0,x_1]\,\alg(x_0)[y] \Bigr)
\nonumber
\\[-0.5ex]
&\hspace{23ex}\text{(where $c = \scriptstyle\sum_{x_0,x_1}\hspace{-0.5ex}\gamma[x_0,x_1]\,\alg(x_1)[y]$)}
\nonumber
\\ =&
\min_{\gamma\in\Gamma}
{\textstyle \sum_{y}}\,c~
f\bigl({\textstyle \sum_{x_0,x_1}}\,
{\textstyle
\frac{\!\gamma[x_0,x_1]\,\alg(x_1)[y] }{\!c }\cdot
\frac{\!\gamma[x_0,x_1]\,\alg(x_0)[y] }{\!\gamma[x_0,x_1]\,\alg(x_1)[y] }
}\bigr)
\nonumber
\\ \le&
\min_{\gamma\in\Gamma}
{\textstyle \sum_{y}}\,c\,
{\textstyle \sum_{x_0,x_1}}\,
{\textstyle
\frac{\!\gamma[x_0,x_1]\,\alg(x_1)[y] }{\!c }
}\cdot
f\bigl(
{\textstyle
\frac{\!\gamma[x_0,x_1]\, \alg(x_0)[y] }{\!\gamma[x_0,x_1]\,\alg(x_1)[y] }
}\bigr)
\nonumber
\\[-0.4ex]
&\hspace{10ex}\text{(by Jensen's inequality and the convexity of $f$)}
\nonumber
\\ =&
\min_{\gamma\in\Gamma}
{\textstyle \sum_{x_0,x_1}}\,
\gamma[x_0,x_1]\,
\sum_{y}\,
\alg(x_1)[y]\,\cdot
\!\displaystyle f\bigl( {\textstyle \frac{\alg(x_0)[y] }{\alg(x_1)[y]}} \bigr)
\nonumber
\\ =&
\min_{\gamma\in\Gamma}
{\textstyle \sum_{x_0,x_1}}\,
\gamma[x_0,x_1]\,
\fdiverge{\alg(x_0)}{\alg(x_1)}
{.}
\label{eq:point:f-divergence}
\end{align}

Assume that $\alg$ provides $(\varepsilon, \Df)$-\DP{} w.r.t. $\varPhi$.
By Definition~\ref{def:lifting-relations}, there is a coupling $\gamma\in\Gamma$ with $\supp(\gamma)\subseteq\varPhi$.
Then:
\begin{align*}
&
\fdiverge{\lift{\alg}(\lambda_0)}{\lift{\alg}(\lambda_1)} 
\nonumber
\\ =&
\min_{\gamma\in\Gamma}
{\textstyle \sum_{x_0,x_1}}\,
\gamma[x_0,x_1]\,
\fdiverge{\alg(x_0)}{\alg(x_1)}
\hspace{4.7ex}
\text{(by \eqref{eq:point:f-divergence})}
\\ \le&
\min_{\gamma\in\Gamma}
{\textstyle \sum_{x_0,x_1}}\,
\gamma[x_0,x_1]\,
\varepsilon
\\[-0.4ex]
&\hspace{5.5ex}
\text{(by $(x_0,x_1)\in\supp(\gamma) \subseteq \varPhi$ and $(\varepsilon, \Df)$-\DP{})}
\\ =&~
\varepsilon
{.}
\end{align*}
Hence $\alg$ provides $(\varepsilon, \Df)$-\DistP{} w.r.t. $\lift{\varPhi}$.
\end{proof}
\vspace{1ex}

\fXDPdist*

\begin{proof}
Assume that $\alg$ provides $(\varepsilon, \utmetric, \Df)$-\XDP{} w.r.t. $\varPhi$.
Let $(\lambda_0, \lambda_1)\in\liftemd{\varPhi}$.
By definition, there exists a coupling $\gamma\in\Gamma$ that satisfies $\supp(\gamma)\subseteq\varPhi$ and $\gamma\in\GammaEMD(\lambda_0, \lambda_1)$.
Then it follows from \eqref{eq:point:f-divergence} in the proof for Theorem~\ref{thm:fDistP} that:
\begin{align*}
&
\fdiverge{\lift{\alg}(\lambda_0)}{\lift{\alg}(\lambda_1)} 
\\ =&
\min_{\gamma\in\Gamma}
{\textstyle \sum_{x_0,x_1}}\,
\gamma[x_0,x_1]\,
\fdiverge{\alg(x_0)}{\alg(x_1)}
\hspace{4ex}
\text{(by \eqref{eq:point:f-divergence})}
\\ \le&
\min_{\gamma\in\Gamma}
{\textstyle \sum_{x_0,x_1}}\,
\gamma[x_0,x_1]\,
\varepsilon\, \utmetric(x_0, x_1)
\\[-0.3ex]
&\hspace{4ex}
\text{(by $(x_0,x_1)\in\supp(\gamma) \subseteq \varPhi$ and $(\varepsilon, \utmetric, \Df)$-\XDP{})}
\\ =&~
\varepsilon\, \sensemd(\lambda_0, \lambda_1)
{.}
\hspace{15ex}
\text{(by $\gamma\in\GammaEMD(\lambda_0, \lambda_1)$)}
\end{align*}
Hence $\alg$ provides $(\varepsilon, \sensemd, \Df)$-\XDistP{} w.r.t. $\liftemd{\varPhi}$.
\end{proof}

\subsection{Point Obfuscation by Distribution Obfuscation}
\label{sub:proofs:DistP-implies-DP}

Next we show that divergence \DP{} is an instance of divergence \DistP{} if an adjacency relation includes pairs of point distributions
(i.e., distributions having single points with probability $1$).
\vspace{1ex}

\arxiv{
\begin{lemma}\label{lem:lifting-includes-point-dist}
Let $p\in\realsone\cup\{\infty\}$ and $\varPhi\subseteq\calx\times\calx$.
For any $(x_0, x_1)\in\varPhi$, we have $(\eta_{x_0}, \eta_{x_1})\in\liftp{\varPhi}$.
\end{lemma}
\vspace{1ex}
}

\begin{restatable}[\DistP{} $\Rightarrow$ \DP{} and \XDistP{} $\Rightarrow$ \XDP{}]{thm}{DPfromDistP}
\label{thm:DPfromDistP}
Let $\varepsilon\in\realsnng$,\, $p\in\realsone\cup\{\infty\}$,\, $D\in\Div{\caly}$, $\varPhi\subseteq\calx\times\calx$, and $\alg:\calx\rightarrow\Dists\caly$ be a randomized algorithm.
\begin{enumerate}
\item
If $\alg$ provides $(\varepsilon, D)$-\DistP{} w.r.t. $\lift{\varPhi}$, then
it provides $(\varepsilon, D)$-\DP{} w.r.t. $\varPhi$.
\item
If $\alg$ provides $(\varepsilon, \Wpu, D)$-\XDistP{} w.r.t. $\liftp{\varPhi}$, then
it provides $(\varepsilon, \utmetric, D)$-\XDP{} w.r.t. $\varPhi$.
\end{enumerate}
\end{restatable}

\conference{See~\cite{Kawamoto:Allerton19:arXiv} for the proof.}

\arxiv{
\begin{proof}
We show the first claim as follows.
Assume that $\alg$ provides $(\varepsilon, D)$-\DistP{} w.r.t. $\lift{\varPhi}$.
Let $(x_0, x_1)\in\varPhi$, and
$\eta_{x_0}$ and $\eta_{x_1}$ be the point distributions.
By Lemma~\ref{lem:lifting-includes-point-dist} and $\liftp{\varPhi}\subseteq\lift{\varPhi}$, we have $(\eta_{x_0}, \eta_{x_1})\in\lift{\varPhi}$.
By $(\varepsilon, D)$-\DistP{}, 
we obtain
$\diverge{\alg(x_0)\!}{\!\alg(x_1)} =
\diverge{\lift{\alg}(\eta_{x_0})\!}{\!\lift{\alg}(\eta_{x_1})}
\le \varepsilon$.
Hence $\alg$ provides $(\varepsilon, D)$-\DP{} w.r.t.~$\varPhi$.

Next we show the second claim.
Assume that $\alg$ provides $(\varepsilon, \Wpu, D)$-\XDistP{} w.r.t. $\liftp{\varPhi}$.
Let $(x_0, x_1)\in\varPhi$, and
$\eta_{x_0}$ and $\eta_{x_1}$ be the point distributions. 
By Lemma~\ref{lem:lifting-includes-point-dist}, we have $(\eta_{x_0}, \eta_{x_1})\in\liftp{\varPhi}$.
Then we obtain:
\begin{align*}
\diverge{\alg(x_0)\!}{\!\alg(x_1)}
&=\,
\diverge{\lift{\alg}(\eta_{x_0})\!}{\!\lift{\alg}(\eta_{x_1})}
\\
&\le\,
\varepsilon \Wpu(\eta_{x_0}, \eta_{x_1})
~~~\text{(by \XDistP{} of $\alg$)}
\\
&=\,
\varepsilon \utmetric(x_0, x_1)
{,}
\end{align*}
where the last equality follows from the definition of $\Wpu$.
Hence $\alg$ provides $(\varepsilon, \utmetric, D)$-\XDP{} w.r.t.~$\varPhi$.
\end{proof}
}

\arxiv{
\subsection{Privacy and Utility of Coupling Mechanisms}
\label{sub:proof:cp}

Next, we show the privacy of the coupling mechanisms.
\vspace{1ex}

\firstCouplingMaxDistP*
\vspace{1ex}

\begin{proof}
Let $((s_0, \lambda_{s_0}), (s_1, \lambda_{s_1})) \in\varPsi$, and $R\subseteq\caly$.
When $\CP$ is applied to $\lambda_{s_0}$ the output distribution is given by:
\begin{align*}
\!\lift{\CP}\!(s_0,\lambda_{s_0})[R]
&=
\sum_{x \in \calx} \lambda_{s_0}[x] \cdot \textstyle\frac{\gamma_{s_0}[x, R]}{\hlambda_{s_0}[x]}
\\ &\le
e^{\varepsilon}\!\sum_{x \in \calx}\!\gamma_{s_0}[x, R]
\hspace{1ex}
\text{(by $\maxdiverge{\lambda_{s_0}\!}{\!\hlambda_{s_0}} \le \varepsilon$)}
\\ &=
e^{\varepsilon} \mu[R]
{.}
\end{align*}
When $\CP$ is applied to $\lambda_{s_1}$ the output distribution is:
\begin{align*}
\!\lift{\CP}\!(s_1,\lambda_{s_1})[R]
&=
\sum_{x \in \calx} \lambda_{s_1}[x] \cdot \textstyle\frac{\gamma_{s_1}[x, R]}{\hlambda_{s_1}[x]}
\\ &\ge
e^{-\varepsilon}\!\sum_{x \in \calx}\!\gamma_{s_1}[x, R]
\hspace{1ex}
\text{(by $\maxdiverge{\hlambda_{s_1}\!}{\!\lambda_{s_1}} \le \varepsilon$)}
\\ &=
e^{-\varepsilon} \mu[R]
{.}
\end{align*}
Hence $\textstyle\frac{\lift{\CP}(s_0,\lambda_0)[R]}{\lift{\CP}(s_1,\lambda_1)[R]} \le e^{2\varepsilon}$.
Therefore $\CP$ provides $(2\varepsilon, \Dinf)$-\DistP{} w.r.t. $\varPsi$.

Next the KL-divergence is given by:
\begin{align*}
&\eqspace
\KLdiverge{ \lift{\CP}(s_0,\lambda_{s_0}) \!}{\! \lift{\CP}(s_1,\lambda_{s_1}) } 
\\ &=
\sup_{y} \lift{\CP}(s_0,\lambda_{s_0})[y]\cdot\,
\ln \left( \textstyle\frac{ \lift{\CP}(s_0,\lambda_{s_0})[y] }{ \lift{\CP}(s_1,\lambda_{s_1})[y] } \right)
\\ &\le
e^{\varepsilon} \sup_{y} \mu[y]\,
\ln \left( e^{2\varepsilon} \right)
\\ &\le
2\varepsilon\, e^{\varepsilon}
{.}
\end{align*}
Therefore $\CP$ provides $(2\varepsilon\, e^{\varepsilon}, \Df)$-\DistP{} w.r.t. $\varPsi$.

Finally, the $f$-divergence is given by:
\begin{align*}
&\eqspace
\fdiverge{ \lift{\CP}(s_0,\lambda_{s_0}) \!}{\! \lift{\CP}(s_1,\lambda_{s_1}) } 
\\ &=
\sup_{y} \lift{\CP}(s_1,\lambda_{s_1})[y]\cdot
f\left( \textstyle\frac{ \lift{\CP}(s_0,\lambda_{s_0})[y] }{ \lift{\CP}(s_1,\lambda_{s_1})[y] } \right)
\\ &\le
e^{\varepsilon} \sup_{y} \mu[y]
f\left( e^{2\varepsilon} \right)
\\ &\le
e^{\varepsilon} f(e^{2\varepsilon})
{.}
\end{align*}
Therefore $\CP$ provides $(e^{\varepsilon} f( e^{2\varepsilon}), \Df)$-\DistP{} w.r.t. $\varPsi$.
\end{proof}

\subsection{Sequential Composition $\Seq$ with Shared Input}
\label{sub:composition:Seq:details}

We first recall the definition of the sequential composition $\Seq$ with shared input (Fig.~\ref{fig:sequential-shared}) in previous work.

\begin{definition}[Sequential composition $\Seq$]\label{def:seq}\rm
Given two randomized algorithms $\alg_0:\calx\rightarrow\Dists\caly_0$ and $\alg_1:\caly_0\allowbreak\times\calx\rightarrow\Dists\caly_1$, we define the \emph{sequential composition} of $\alg_0$ and $\alg_1$ as the randomized algorithm $\alg_1 \Seq \alg_0: \calx\rightarrow\Dists\caly_1$ such that for any $x\in\calx$,\,
$(\alg_1 \Seq \alg_0)(x) = \alg_1(\alg_0(x), x))$.
\end{definition}
\vspace{1ex}

Then we present the compositionality of $\DKL$-\DistP{}.
Note that since this composition is adaptive, the compositionality does not hold in general for $f$-divergence.

\begin{restatable}[Sequential composition $\Seq$ of $\DKL$-\DistP{}]{prop}{CompositionKL}\label{prop:Composition:KL}
Let 
$\varPhi \subseteq \calx\times\calx$.
If $\alg_0:\calx \rightarrow\Dists\caly_0$ provides $(\varepsilon_0, \DKL)$-\DistP{} w.r.t. $\lift{\varPhi}$ 
and for each $y_0\in\caly_0$, $\alg_1(y_0):\calx \rightarrow\Dists\caly_1$ provides $(\varepsilon_1, \DKL)$-\DistP{} w.r.t. $\lift{\varPhi}$, the sequential composition $\alg_1 \Seq \alg_0$ provides $(\varepsilon_0+\varepsilon_1, \DKL)$-\DistP{} w.r.t.~$\lift{\varPhi}$.
\end{restatable}
\vspace{1ex}

\begin{proof}
By Theorem~\ref{thm:DPfromDistP} in Appendix~\ref{sub:proofs:DistP-implies-DP}, $\alg_0$ provides $(\varepsilon_0, \DKL)$-\DP{} w.r.t. $\varPhi$, 
and for each $y_0\in\caly_0$, $\alg_1(y_0)$ provides $(\varepsilon_1, \DKL)$-\DP{} w.r.t. $\varPhi$.
Let 
$(x, x')\in\varPhi$. Then:
\begin{align*}
&\eqspace
\KLdiverge{(\alg_1 \Seq \alg_0)(x)}{(\alg_1 \Seq \alg_0)(x')}
\\ &=
\sum_{y_1}
(\alg_1 \Seq \alg_0)(x)[y_1]
\ln{\textstyle\frac{ (\alg_1 \Seq \alg_0)(x)[y_1] }{ (\alg_1 \Seq \alg_0)(x')[y_1] }}
\\ &=
\sum_{y_0,y_1}
\alg_0(x)[y_0]\cdot\alg_1(y_0,x)[y_1]
\ln
{\textstyle\frac{\alg_0(x)[y_0]\cdot\alg_1(y_0,x)[y_1]}
     {\alg_0(x')[y_0]\cdot\alg_1(y_0,x')[y_1]}}
\\ &=
\sum_{y_0}
\alg_0(x)[y_0]
\ln
{\textstyle\frac{\alg_0(x)[y_0]}{\alg_0(x')[y_0]}}
\\[0.5ex]
&\eqspace +
\sum_{y_0,y_1}
\alg_0(x)[y_0]
\alg_1(y_0, x)[y_1]
\ln
{\textstyle\frac{\alg_1(y_0,x)[y_1]}{\alg_1(y_0,x')[y_1]}}
\\ &\le
\KLdiverge{\alg_0(x)\!}{\!\alg_0(x')}
\\[-0.5ex]
&\eqspace 
+
\max_{y_0}\!
\sum_{y_1}\!
\alg_1(y_0,x)[y_1]
\ln
{\textstyle\frac{\alg_1(y_0,x)[y_1]}
     {\alg_1(y_0,x')[y_1]}}
\\ &=
\KLdiverge{\!\alg_0(x)\!}{\!\alg_0(x')}
+
\max_{y_0}\KLdiverge{\!\alg_1(y_0,x)\!}{\!\alg_1(y_0,x')}
\\[-0.5ex] &\le
\varepsilon_0 + \varepsilon_1
{.}
\end{align*}
Hence $\alg_1 \Seq \alg_0$ provides $(\varepsilon_0+\varepsilon_1, \DKL)$-\DP{} w.r.t. $\varPhi$.
By Theorem~\ref{thm:fDistP}, $\alg_1 \Seq \alg_0$ provides $(\varepsilon_0+\varepsilon_1, \DKL)$-\DistP{} w.r.t.~$\lift{\varPhi}$.
\end{proof}
\vspace{1ex}

\begin{restatable}[Sequential composition $\Seq$ of $\DKL$-\XDistP{}]{prop}{CompositionKLXDistP}\label{prop:Composition:KL:XDistP}
Let 
$\utmetric$ be a metric over $\calx$, and
$\varPhi\subseteq\calx\times\calx$.
If $\alg_0:\calx\rightarrow\Dists\caly_0$ provides $(\varepsilon_0,  \sensemd, \DKL)$-\XDistP{} w.r.t. $\liftemd{\varPhi}$ 
and for each $y_0\in\caly_0$, $\alg_1(y_0): \calx\rightarrow\Dists\caly_1$ provides $(\varepsilon_1, \sensemd, \DKL)$-\XDistP{} w.r.t. $\liftemd{\varPhi}$ then the sequential composition $\alg_1 \Seq \alg_0$ provides $(\varepsilon_0+\varepsilon_1, \sensemd, \allowbreak \DKL)$-\XDistP{} w.r.t.~$\liftemd{\varPhi}$.
\end{restatable}
\vspace{1ex}

\begin{proof}
Analogous to the proof 
for Proposition~\ref{prop:Composition:KL}.
\end{proof}

\subsection{Sequential Composition $\liftSeq$ with Independent Sampling}
\label{sub:composition:liftSeq:details}

In this section we present the compositionality with independent sampling,
which is defined as follows.

\begin{definition}[Sequential composition $\liftSeq$]\label{def:seq:lift}\rm
Given two randomized algorithms $\alg_0:\calx\rightarrow\Dists\caly_0$ and $\alg_1:\caly_0\times\calx\rightarrow\Dists\caly_1$, we define the \emph{sequential composition} of $\alg_0$ and $\alg_1$ as the randomized algorithm $\alg_1 \liftSeq \alg_0: \calx\times\calx\rightarrow\Dists\caly_1$ such that: for any $x_0, x_1\in\calx$,\,
$(\alg_1 \liftSeq \alg_0)(x_0, x_1) = \alg_1(\alg_0(x_0), x_1))$.
\end{definition}
\vspace{1ex}
We define an operator $\diamond$ between binary relations $\varPsi_0$ and~$\varPsi_1$:
\[
\varPsi_0 \mathbin{\diamond} \varPsi_1 =
\{ (\lambda_0\times\lambda_1, \lambda'_0\times\lambda'_1) \,|\,
   (\lambda_0,\lambda'_0)\in\varPsi_0, (\lambda_1,\lambda'_1)\in\varPsi_1 \}.
\]

Now we show the compositionality for $\DKL$-\DistP{}.

\begin{restatable}[Sequential composition $\liftSeq$ of $\DKL$-\DistP{}]{prop}{CompositionKLlift}\label{prop:Composition:KLlift}
Let 
$\varPsi \subseteq \Dists\calx\times\Dists\calx$.
If $\alg_0:\calx \rightarrow\Dists\caly_0$ provides $(\varepsilon_0, \DKL)$-\DistP{} w.r.t. $\varPsi$ 
and for each $y_0\in\caly_0$, $\alg_1(y_0):\calx \rightarrow\Dists\caly_1$ provides $(\varepsilon_1, \DKL)$-\DistP{} w.r.t. $\varPsi$,
then the composition $\alg_1 \liftSeq \alg_0$ provides $(\varepsilon_0+\varepsilon_1, \DKL)$-\DistP{} w.r.t. $\varPsi \diamond \varPsi$.
\end{restatable}

\begin{proof}
Let 
$(\lambda_0,\lambda'_0), (\lambda_1,\lambda'_1)\in\varPsi$.
\begin{align*}
&\, \KLdiverge{\lift{(\alg_1 \liftSeq \alg_0)}(\lambda_0\times\lambda_1)}{\lift{(\alg_1 \liftSeq \alg_0)}(\lambda'_0\times\lambda'_1)}
\\ =&
\sum_{y_1}
\lift{(\alg_1 \liftSeq \alg_0)}(\lambda_0\times\lambda_1)[y_1]
\ln{\textstyle\frac{ \lift{(\alg_1 \liftSeq \alg_0)}(\lambda_0\times\lambda_1)[y_1] }
        { \lift{(\alg_1 \liftSeq \alg_0)}(\lambda'_0\times\lambda'_1)[y_1] } }
\\ =&
\sum_{y_0,y_1\!}\!%
\lift{\alg_0\!}\!(\!\lambda_0\!)[y_0]
\lift{\alg_1(y_0)\!}\!(\!\lambda_1\!)[y_1]
\ln\!
{\textstyle\frac{\lift{\alg_0\!}\!(\lambda_0)[y_0]\lift{\alg_1(y_0)\!}\!(\lambda_1)[y_1]}
     {\lift{\alg_0\!}\!(\lambda'_0)[y_0]\lift{\alg_1(y_0)\!}\!(\lambda'_1)[y_1]}}
\\ =&
\sum_{y_0}
\lift{\alg_0}(\lambda_0)[y_0]
\ln
{\textstyle \frac{\lift{\alg_0}(\lambda_0)[y_0]}
     {\lift{\alg_0}(\lambda'_0)[y_0]} }
\\[-0.5ex] &
+
\sum_{y_0,y_1}
\lift{\alg_0}(\lambda_0)[y_0]
\lift{\alg_1(y_0)}(\lambda_1)[y_1]
\ln
{\textstyle \frac{\lift{\alg_1(y_0)}(\lambda_1)[y_1]}
     {\lift{\alg_1(y_0)}(\lambda'_1)[y_1]} }
\\ \le&
\KLdiverge{\lift{\alg_0}(\lambda_0)}{\lift{\alg_0}(\lambda'_0)}
\\[-0.5ex]&
+
\max_{y_0}
\sum_{y_1}
\lift{\alg_1}(y_0)(\lambda_1)[y_1]
\ln
{\textstyle \frac{\lift{\alg_1}(y_0)(\lambda_1)[y_1]}
     {\lift{\alg_1}(y_0)(\lambda'_1)[y_1]} }
\\ =&\,
\KLdiverge{\lift{\alg_0}(\lambda_0)}{\lift{\alg_0}(\lambda'_0)}
\\[-0.5ex] &
+
\max_{y_0}\KLdiverge{\lift{\alg_1}(y_0)(\lambda_1)}{\lift{\alg_1}(y_0)(\lambda'_1)}
\\[-0.5ex] \le&\,
\varepsilon_0 + \varepsilon_1
{.}
\end{align*}
Hence $\alg_1 \liftSeq \alg_0$ provides $(\varepsilon_0+\varepsilon_1, \DKL)$-\DistP{} w.r.t. $\varPsi \diamond \varPsi$.
\end{proof}

\begin{restatable}[Sequential composition $\liftSeq$ of $\DKL$-\XDistP{}]{prop}{CompositionKLliftX}\label{prop:Composition:KLliftX}
Let 
$\utmetric$ be a metric over $\calx$, and
$\varPsi \subseteq \Dists\calx\times\Dists\calx$.
If $\alg_0:\calx \rightarrow\Dists\caly_0$ provides $(\varepsilon_0, \sensemd, \DKL)$-\XDistP{} w.r.t. $\varPsi$ 
and for each $y_0\in\caly_0$, $\alg_1(y_0):\calx \rightarrow\Dists\caly_1$ provides $(\varepsilon_1, \sensemd, \DKL)$-\XDistP{} w.r.t. $\varPsi$,
then the composition $\alg_1 \liftSeq \alg_0$ provides $(\varepsilon_0+\varepsilon_1, \sensemd, \DKL)$-\XDistP{} w.r.t. $\varPsi \diamond \varPsi$.
\end{restatable}

\begin{proof}
Analogous to the proof 
for Proposition~\ref{prop:Composition:KLlift}.
\end{proof}

\subsection{Post-processing and Pre-processing}
\label{sub:post-pre-processing}

Next we show that divergence distribution privacy is immune to the post-processing.
For $\alg_0:\calx\rightarrow\Dists\caly$ and $\alg_1:\caly\rightarrow\Dists\calz$, we define $\alg_1 \circ \alg_0$ by: $(\alg_1 \circ \alg_0)(x) = \alg_1(\alg_0(x))$.

\begin{restatable}[Post-processing]{prop}{PostProcess}\label{prop:PostProcess}
Let 
$\varPsi\subseteq\Dists\calx\times\Dists\calx$, and $\sensfunc:\Dists\calx\times\Dists\calx\rightarrow\realsnng$ be a metric.
Let $\alg_0:\calx\rightarrow\Dists\caly$ and $\alg_1:\caly\rightarrow\Dists\calz$.
\begin{enumerate}
\item
If $\alg_0$ provides $(\varepsilon, \Df)$-\DistP{} w.r.t. $\varPsi$ then so does the composite function $\alg_1\circ\alg_0$.
\item
If $\alg_0$ provides $(\varepsilon, \sensfunc, \Df)$-\XDistP{} w.r.t. $\varPsi$ then so does the composite function $\alg_1\circ\alg_0$.
\end{enumerate}
\end{restatable}
\vspace{0.5ex}

\begin{proof}
The claim is immediate from the data processing inequality for the $f$-divergence.
\end{proof}
\vspace{0.5ex}

We then show properties of pre-processing as follows.
\begin{restatable}[Pre-processing]{prop}{PreProcessMAX}\label{prop:PreProcessMAX}
Let $c\in\realsnng$, 
$\varPsi \subseteq \Dists\calx\times\Dists\calx$,
$\sensfunc:\Dists\calx\times\Dists\calx\rightarrow\realsnng$ be a metric, and $D\in\Div{\caly}$.
\begin{enumerate}
\item
If $T:\Dists\calx\rightarrow\Dists\calx$ is a \emph{$(c,\varPsi)$-stable} transformation and 
$\alg:\calx\rightarrow\Dists\caly$ provides $(\varepsilon,D)$-\DistP{} w.r.t. $\varPsi$, then $\alg\circ T$ provides $(c\,\varepsilon, D)$-\DistP{} w.r.t. $\varPsi$.
\item
If $T:\Dists\calx\rightarrow\Dists\calx$ is a \emph{$(c,\sensfunc)$-stable} transformation and 
$\alg:\calx\rightarrow\Dists\caly$ provides $(\varepsilon,\sensfunc, D)$-\XDistP{}, then $\alg\circ T$ provides $(c\,\varepsilon, \sensfunc, D)$-\XDistP{}.
\end{enumerate}
\end{restatable}
\vspace{0.5ex}

\begin{proof}
We show the first claim as follows.
Assume that $\alg$ provides $(\varepsilon,D)$-\DistP{} w.r.t. $\varPsi$.
Let $(\lambda,\lambda')\in\varPsi$.
Then
$\diverge{\lift{(\alg\circ T)}(\lambda)\!}{\!\lift{(\alg\circ T)}(\lambda')} =
\diverge{\lift{\alg}(\lift{T}(\lambda))\!}{\!\lift{\alg}(\lift{T}(\lambda'))} 
\le c \varepsilon$
by $(c,\varPsi)$-stability.
Therefore $\alg\circ T$ provides $(c\,\varepsilon, D)$-\DistP{} w.r.t. $\varPsi$.

Next we show the second claim.
Assume that $\alg$ provides $(\varepsilon,\sensfunc,D)$-\XDistP{}.
Let $\lambda,\lambda'\in\Dists\calx$.
Then we obtain:
\begin{align*}
&\eqspace
\diverge{\lift{(\alg\circ T)}(\lambda)\!}{\!\lift{(\alg\circ T)}(\lambda')}
\\ &=
\diverge{\lift{\alg}(\lift{T}(\lambda))\!}{\!\lift{\alg}(\lift{T}(\lambda'))}
\\ &\le
\varepsilon\sensfunc(\lift{T}(\lambda),\lift{T}(\lambda'))
\\ &\le
c\,\varepsilon\sensfunc(\lambda,\lambda')
\hspace{3ex}\text{(by $(c,\sensfunc)$-stable)}
{.}
\end{align*}
Therefore $\alg\circ T$ provides $(c\,\varepsilon, \sensfunc, D)$-\XDistP{}.
\end{proof}

\subsection{Relationships among \XDistP{} Notions}
\label{subsec:relationships:notions}

Finally, we show relationships among distribution privacy notions with different metric $d$ and divergence $D$.
\vspace{0.5ex}

\begin{restatable}[$\sensemd$-\XDistP{} \,$\Rightarrow$ $\sensinf$-\XDistP{}]{prop}{relationWDistP}
\label{prop:relationWDistP}
Let 
$D\in\Div{\caly}$.
If $\alg:\calx\rightarrow\Dists\caly$
provides $(\varepsilon, \sensemd, \mathit{D})$-\XDistP{}, then it provides $(\varepsilon, \sensinf, \mathit{D})$-\XDistP{}.
\end{restatable}
\vspace{0.5ex}

\begin{proof}
Assume that $\alg$ provides $(\varepsilon, \sensemd, \mathit{D})$-\XDistP{}.
Let $\lambda_0,\lambda_1\in\Dists\calx$.
By the property of the $p$-Wasserstein metric,
$\sensemd(\lambda_0,\lambda_1) \le \sensinf(\lambda_0,\lambda_1)$.
Then
$\mathit{D}(\mu_0\parallel\mu_1) \le
 \sensemd(\lambda_0,\lambda_1) \le
 \sensinf(\lambda_0,\lambda_1)$.
Hence the claim follows.
\end{proof}
\vspace{0.5ex}

\begin{restatable}[$\mathit{D}\le\mathit{D'}$ \& $\mathit{D'}$-\XDistP{} \,$\Rightarrow$ $\mathit{D}$-\XDistP{}]{prop}{relationDDistP}
\label{prop:relationDDistP}
Let 
$d: (\Dists\calx \times \Dists\calx)\rightarrow\reals$ be a metric.
Let $D, D'\in\Div{\caly}$ be two divergences such that for all $\mu_0, \mu_1\in\Dists\caly$,\,
$\diverge{\mu_0}{\mu_1} \le \mathit{D'}(\mu_0\parallel\mu_1)$.
If $\alg:\calx\rightarrow\Dists\caly$ provides
$(\varepsilon, d, \mathit{D'})$-\XDistP{}, then it provides $(\varepsilon, d, \mathit{D})$-\XDistP{}.
\end{restatable}
Then $(\varepsilon, d, \Dinf)$-\XDistP{} implies $(\varepsilon, d, \DKL)$-\XDistP{}.
\vspace{0.5ex}

\begin{proof}
Assume $\alg$ provides $(\varepsilon, d, \mathit{D'})$-\XDistP{}.
Let $\lambda_0, \lambda_1\in\Dists\calx$.
Then
$
\mathit{D'}(\lift{\alg}(\lambda_0){\parallel}\lift{\alg}(\lambda_1))
\le \varepsilon d(\lambda_0,\lambda_1)
{.}
$
By definition,
$\mathit{D}(\lift{\alg}(\lambda_0){\parallel}\lift{\alg}(\lambda_1)) \le
 \mathit{D'}(\lift{\alg}(\lambda_0){\parallel}\lift{\alg}(\lambda_1))\allowbreak \le
 \varepsilon d(\lambda_0,\lambda_1)$.
Thus $\alg$ provides $(\varepsilon, d, \mathit{D})$-\XDistP{}.
\end{proof}

}

\end{document}